%% file: main.tex
\documentclass[sigconf, nonacm]{acmart}

\newcommand\vldbdoi{10.14778/3574245.3574267}
\newcommand\vldbpages{842 - 855}
\newcommand\vldbvolume{16}
\newcommand\vldbissue{4}
\newcommand\vldbyear{2022}
\newcommand\vldbauthors{\authors}
\newcommand\vldbtitle{\shorttitle} 
\newcommand\vldbavailabilityurl{https://github.com/brucechin/dodb}
\newcommand\vldbpagestyle{empty} 

\usepackage{amsmath}
\usepackage{amsthm}
\usepackage{graphicx}
\usepackage{array}
\usepackage{float}
\usepackage{subfig}
\usepackage{algorithm}
\usepackage{algpseudocode}
\usepackage{tikz}
\usepackage{comment}
\usepackage{hyperref}
\usepackage{multirow}
\hypersetup{colorlinks=true,citecolor=red,linkcolor=blue}
\usepackage[makeroom]{cancel}
\usepackage{stmaryrd}
\usepackage{booktabs}
\usepackage{makecell}
\usepackage{pifont}
\usepackage{wrapfig}
\usepackage{lipsum}
\usepackage{listings}
\usepackage{mathtools}
\usepackage{xspace}
\usepackage{balance}

\usepackage{adjustbox}

\graphicspath{{./figs/}}

\definecolor{b2}{RGB}{51,153,255}
\definecolor{mygreen}{RGB}{80,180,0}

\definecolor{yl}{RGB}{255,80,0}

\definecolor{thegreen}{RGB}{40,100,0}
\definecolor{theblue}{RGB}{0,0,204}

\theoremstyle{plain}
\newtheorem{theorem}{Theorem}[section]
\newtheorem{lemma}[theorem]{Lemma}

\newtheorem{fact}[theorem]{Fact}
\newtheorem{claim}[theorem]{Claim}

\theoremstyle{definition}
\newtheorem{definition}[theorem]{Definition}

\newtheorem{remark}[theorem]{Remark}

\renewcommand{\tilde}{\widetilde}

\DeclareMathOperator{\poly}{poly}

\DeclareMathOperator{\R}{{\mathbb R}}
\DeclareMathOperator*{\E}{{\mathbb{E}}}

\DeclarePairedDelimiter{\ceil}{\lceil}{\rceil}
\newcommand{\eps}{\epsilon}
\newcommand{\ex}[1]{\mathbb{E}\normalfont\lbrack #1 \rbrack}

\newcommand{\LHS}{\mathrm{LHS}}

\newcommand{\group}[2]{\gamma_{#1}{#2}}
\newcommand{\cmark}{\ding{51}}%
\newcommand{\code}[1]{{\textsc{#1}}}

\newcommand{\sys}{\textsc{Adore}\xspace}
\newcommand{\filteralg}{{\textsc{DoFilter}}}
\newcommand{\groupsalg}{{\textsc{DoGroup}$_s$}}
\newcommand{\grouphalg}{{\textsc{DoGroup}$_h$}}

\newcommand{\ignore}[1]{}

\algnewcommand\algorithmicforeach{\textbf{for each}}
\algdef{S}[FOR]{ForEach}[1]{\algorithmicforeach\ #1\ \algorithmicdo}

\newcommand{\lianke}[1]{{\color{orange}[Lianke: #1]}}

\newcommand{\vldbrevision}[1]{{\color{black} #1}}

\def\Snospace~{\S{}}

\title{Adore: Differentially Oblivious Relational Database Operators}

\date{}

\author{Lianke Qin}
\affiliation{
\institution{UC Santa Barbara}
}
\email{lianke@ucsb.edu}

\author{Rajesh Jayaram}
\affiliation{
\institution{Carnegie Mellon University}
}
\email{rkjayara@cs.cmu.edu}

\author{Elaine Shi}
\affiliation{
\institution{Carnegie Mellon University}
}
\email{runting@cs.cmu.edu}

\author{Zhao Song}
\affiliation{
\institution{Adobe Research}
 }
\email{zsong@adobe.com}

\author{Danyang Zhuo}
\affiliation{
\institution{Duke University}
}
\email{danyang@cs.duke.edu}

\author{Shumo Chu}
\affiliation{
\institution{p0x labs}
}
\email{chushumo@cs.washington.edu}

\begin{document}

\begin{abstract}
\input{abstract}

\end{abstract}

\maketitle
\pagestyle{\vldbpagestyle}
\begingroup\small\noindent\raggedright\textbf{PVLDB Reference Format:}\\
\vldbauthors. \vldbtitle. PVLDB, \vldbvolume(\vldbissue): \vldbpages, \vldbyear.\\
\href{https://doi.org/\vldbdoi}{doi:\vldbdoi}
\endgroup
\begingroup
\renewcommand\thefootnote{}\footnote{\noindent
This work is licensed under the Creative Commons BY-NC-ND 4.0 International License. Visit \url{https://creativecommons.org/licenses/by-nc-nd/4.0/} to view a copy of this license. For any use beyond those covered by this license, obtain permission by emailing \href{mailto:info@vldb.org}{info@vldb.org}. Copyright is held by the owner/author(s). Publication rights licensed to the VLDB Endowment. \\
\raggedright Proceedings of the VLDB Endowment, Vol. \vldbvolume, No. \vldbissue\ %
ISSN 2150-8097. \\
\href{https://doi.org/\vldbdoi}{doi:\vldbdoi} \\
}\addtocounter{footnote}{-1}\endgroup

\ifdefempty{\vldbavailabilityurl}{}{
\vspace{.3cm}
\begingroup\small\noindent\raggedright\textbf{PVLDB Artifact Availability:}\\
The source code, data, and/or other artifacts have been made available at \url{https://github.com/brucechin/dodb}.
\endgroup
}

\balance

\input{intro}
\input{sgx}

\input{background}

\input{operators}

\input{dp_distinct}

\input{exp}

\input{relat}

\input{discussion}
\input{conclusion}
\newpage
\bibliographystyle{ACM-Reference-Format}
%
\bibliography{ref}

\newpage
\appendix 

\cleardoublepage
\input{app}
\input{appendix_dp_distinct}
\input{rand_partition}

\end{document}

%% file: abstract.tex
There has been a recent effort in applying differential privacy on memory access patterns to enhance data privacy. This is called differential obliviousness. Differential obliviousness is a promising direction because it provides a principled trade-off between performance and desired level of privacy. To date, it is still an open question whether differential obliviousness can speed up database processing with respect to full obliviousness. In this paper, we present the design and implementation of \textbf{Adore}: \textbf{A} set of \textbf{D}ifferentially \textbf{O}blivious \textbf{RE}lational database operators. Adore includes selection with projection, grouping with aggregation, and foreign key join. We prove that they satisfy the notion of differential obliviousness. 
Our differentially oblivious operators have reduced cache complexity, runtime complexity, and output size compared to their state-of-the-art fully oblivious counterparts. We also demonstrate that our implementation of these differentially oblivious operators can outperform their state-of-the-art fully oblivious counterparts by up to $7.4\times$.


%% file: intro.tex
\section{Introduction}

Moving data and computation to the cloud is the most dominant trend in the industry today. Cloud databases~\cite{bigdata1,bigdata2, bigdata3, bigdata4} collect and analyze a vast amount of user data, including sensitive information
such as health data, financial records, and social interactions.
These databases allow developers to run complex queries using a SQL interface, the de facto standard for data analytics. 
Because of these developments, cloud data is often the central target of attacks
\cite{hack1, hack2, hack3, hack4, hack5}, protecting sensitive data in cloud databases has become more important than ever.

A promising direction is to use hardware enclaves, such as Intel SGX ~\cite{SGX-enclave},
and RISC-V Sanctum~\cite{enclave2}, to provide secure data processing inside the cloud.
These enclaves are protected regions in CPUs, where a remotely attested 
piece of code can run without interference from a potentially adversarial
hypervisor and OS.
Major processor vendors have all equipped their new generation of CPUs 
with hardware enclaves.
Cloud providers like Microsoft and Alibaba provide enclave 
support in their public cloud offerings \cite{azureconfidential, alibabasgx}.
Some cloud databases~\cite{Azure-SQL-encrypted, azuresgx} have already used Intel SGX to protect user data, and it is also an area of active research~\cite{priebe2018enclavedb, cipherbase}.

Unfortunately, the Achilles' heel of using hardware enclaves is that enclaves alone do not protect the access patterns of encrypted data outside the enclave's memory.
For applications like big data analytics that require managing a large 
amount of data, an enclave has to fetch encrypted data residing outside the enclave (e.g., a server's main memory, disks).
This leads to {\em access pattern} attacks~\cite{pattern1, pattern2}.
A long list 
of practical access pattern attacks of this form~\cite{attack1, attack2, attack3, attack4, attack5, kkno16} have been discovered for encrypted databases such as CryptDB~\cite{popa2011cryptdb} and TrustedDB~\cite{bajaj2013trusteddb}.

One approach to address this vulnerability is to 
make the memory access patterns of enclave-based database systems 
{\em oblivious}, which means that the access patterns of the system are indistinguishable for different input data. 
This notion of obliviousness was first proposed by Goldreich and Ostrovsky~\cite{oram}.
However, making the database systems fully oblivious incurs a huge performance penalty.
For example, any query output including intermediate results must be padded with filler tuples to the worse-case size, which is usually much larger than the actual result size. 
In recent enclave-based databases (e.g., Opaque~\cite{zheng2017opaque}, ObliDB~\cite{eskandarian2017oblidb}), their fully oblivious modes\footnote{ObliDB calls its fully oblivious mode padding mode.} are significantly slower than their partially oblivious or non-oblivious counterparts.
While their partially oblivious or non-oblivious mode either does not protect memory access patterns at all or has arbitrary leakage, such as leaking the sizes of intermediate results and outputs. The ramifications of such leakages are not understood and may likely lead to new attacks.

\begin{table*}[ht]
\small
\centering
\caption[caption for table 1]{ Cache complexity/Private Memory Size/Output size  comparisons of our system with ObliDB ~\cite{eskandarian2017oblidb}. $N$ denotes input size, $R$
denotes output size, $B$ denotes block size, $M$ denotes
the size of private memory. Let FO denote Full obliviousness. Let DO denote Differential obliviousness. Let * denote the result from our interpretation of their algorithms, the original work didn't explicit state the result. Let \cmark denote the best choice. \vldbrevision{ObliDB's hash-based grouping with aggregation fails when the total number of distinct groups exceeds the enclave private memory capacity $M$. Cache complexity measures the total numbers of data blocks movement between private memory and untrusted memory, which is the dominant overhead. Runtime complexity measures the total number of CPU instructions executed on the decrypted data within enclave.}
}
\begin{adjustbox}{width=\textwidth}

\begin{tabular}{|l|l|c|l|l|l|l|} \hline
 {\bf Operator} & {\bf Algorithms} & {\bf Privacy} & {\bf Private Mem. Size} & {\bf Cache Complexity} & {\vldbrevision{\bf Runtime Complexity}} &  {\bf Output Size} \\ \hline
 Selection with projection &  ObliDB & FO & 1 \cmark  & $2N$*  & \vldbrevision{$O(N)$*} & $N$   \\ 
 Selection with projection& {\bf Alg. \ref{alg:dofilter}} & DO & $\poly\log (N)$ & $(N+R)/B$ \cmark & \vldbrevision{$O(N + R)$} & $R + $ $\poly\log(N)$ \cmark  \\ \hline
 Grouping with aggregation & ObliDB & FO & $M (M > R) $& $2N$* & \vldbrevision{$O(N)$*  } & $M $* \\
 Grouping with aggregation &{\bf Alg. \ref{alg:dogrouph}} & DO & $M$ ($M \geq O(\eps^{-1}\log^2(1/\delta))$) & $N/B + 11 NR/9MB$  \cmark & \vldbrevision{$O(NR/M)$} & $\frac{11}{9}R$ \\ \hline
 Foreign key join & ObliDB & FO & $\poly\log(N)$ \cmark  & $N\cdot\log^2(N)$*  & \vldbrevision{$O(N \log^2(N))$*} &  $N$  \\ 
 Foreign key join &{\bf Alg. \ref{alg:dojoin}} & DO & $\poly\log(N)$\cmark & $6(N/B)\cdot\log(N/B) + N/B + R/B + \poly\log(N)/B$ \cmark & \vldbrevision{$O(N \log(N))$ \cmark}  & $R + \poly\log(N) $ \cmark \\ \hline
\end{tabular}
\end{adjustbox}

\label{tab:complexity}
\end{table*}

Recently, there are rising interests in adopting differential privacy to protect access pattern leakage. To apply this idea to databases, instead of making access patterns indistinguishable between all inputs, \emph{we make the access patterns satisfy differential privacy}~\cite{dmns06, dr14}, a privacy model that only requires indistinguishability among neighboring databases. 
This notion is called \emph{differential obliviousness}, and was introduced by Chan et al.~\cite{ccms19}. This relaxation from full obliviousness opens up new design spaces for more efficient algorithms, yet still provides provable
privacy guarantees for each database record. 
Differential oblivious algorithms only add dummy reads/writes during execution to obfuscate the memory access pattern, and they do not change the query results (except additional dummy tuples). The query accuracy is therefore not affected.

This raises two salient questions: (1) How to design differentially oblivious database operators? (2) Can differentially oblivious database operators outperform their state-of-the-art fully oblivious operators?

In this paper, we present \textbf{Adore}: \textbf{A} set of \textbf{D}ifferentially \textbf{O}blivious \textbf{RE}lational database operators, including selection with projection,  grouping with aggregation, and foreign key join. We pick these operators because they are sufficient to support important database workloads, such as big data benchmark (BDB) \cite{bdb}. We use three key theoretical performance metrics to guide our design: (1) cache complexity, \vldbrevision{(2) runtime complexity,} and (3) output size. Cache complexity measures the total numbers of blocks read from untrusted memory to enclave memory (a.k.a. {\em private memory}), and written from private memory to untrusted memory.  In this scenario, the enclave memory is the ``cache'' and each page is a ``block'' (i.e., the atomic unit being swapped in and out). 
Cache complexity is a dominant source of query latency because moving data between trusted and untrusted memory requires \vldbrevision{memory copying, encryption, decryption and the overhead from SGX ECALLS and OCALLS}.
Using this metric is further justified by our microbenchmark results in~\autoref{sec:breakdown}: in most queries, the memory copy, encryption, and decryption together constitute more than $80\% $ of total query completion time. 
\vldbrevision{Second, we consider runtime complexity, which is the total number of CPU instructions executed on the decrypted data within enclave to implement the filter, aggregation and join operators, but they are not the dominant overhead compared with data movement between enclave private memory and untrusted memory which is measured by cache complexity metric.} 
Output size is also an important metric: the output size decides how much data needs to transfer from trusted to untrusted memory to generate the output.
\vldbrevision{In Table~\ref{tab:complexity}, we provide detailed comparison between our differentially oblivious database operators and ObliDB. We can reduce the cache complexity and output size over ObliDB asymptotically. }

We implement these differentially oblivious database operators. 
Our prototype is developed on top of Intel SGX, because of its availability. We acknowledge that SGX will be deprecated starting from the 11th generation of Intel CPU Core CPUs, but it will continue on Intel Xeon CPUs for cloud usage, which is our target deployment scenario. Choosing SGX also means our prototype is susceptible to known SGX vulnerabilities. These vulnerabilities have known solutions, and patching our prototypes for these vulnerabilities is out of scope of our paper. (See \autoref{sec:discussion}.)

We evaluate our implementation using workloads from BDB.
We show that our operators can substantially outperforms the fully oblivious operators in ObliDB.
Overall, our operators provide up to $7.4\times$ performance improvement over ObliDB.
Our operators also allow scaling to larger data compared with the existing oblivious operators: our operators can process input tables containing 30 million tuples in groupby in BDB, while ObliDB fails because the total number of distinct groups is larger than $M$. Our source code and scripts for running the evaluation are available anonymously at \url{https://github.com/brucechin/dodb}.

Our paper makes the following contributions:
\begin{itemize}
    \item We apply the notion of differential obliviousness to database operators to enhance data privacy by designing three new differentially oblivious database operators.
    \item We formally prove that these operators satisfy the notion of differential obliviousness and have reduced cache complexity and output size.
    \item We demonstrate empirically these differential oblivious operators' performance gain compared to their state-of-the-art fully oblivious counterparts.
\end{itemize}

\vldbrevision{\textbf{Roadmap.} We first introduce our threat model and background knowledge in \autoref{sec:background}. We present our differentially oblivious operators including filter, grouping with aggregation and foreign key join in \autoref{sec:ops}. We present a differentially private distinct count algorithm in \autoref{sec:dp_distinct} and use it in differentially oblivious grouping with aggregation. We evaluate the performance improvement of our algorithm in \autoref{sec:eval}.
We discuss the related work in \autoref{sec:related_work}.
We discuss the potential future work in \autoref{sec:discussion}.
We conclude our paper in \autoref{sec:conclusion}.}

%% file: sgx.tex
\section{Background}\label{sec:background}
In this section, we first describe our threat model (\autoref{sec:threat}). Then, we formally define differential obliviousness and compare
it with full obliviousness (\autoref{sec:models}).

\subsection{Threat Model}
\label{sec:threat}

We use Intel SGX as an example to discuss the threat model.
Intel SGX provides confidentiality and integrity of its enclave memory (i.e., \emph{private memory}), which is located in a preconfigured part of DRAM called the Processor Reserved Memory (PRM). The content in the enclave memory is encrypted. The enclave memory also guarantees integrity: only the code residing inside the enclave can modify the enclave memory after the enclave is created. The enclave memory size has an upper bound (i.e., 128\,MB). An SGX enclave has a predefined entry point, so a user process or the OS cannot invoke the enclave to run at arbitrary memory addresses.
SGX provides \emph{remote attestation} to allow a remote system to verify what code is loaded into an enclave, and set up a secure communication channel to the enclave. 

These SGX features allow us to trust the code running inside the enclave. Untrusted processes and the operating system cannot tamper with the database source code inside the enclave. The execution and memory accesses for the private memory are also invisible to the untrusted processes and the operating system.

However, the database requires an untrusted component for I/O.
For a trusted data owner to use the database, the data owner sends an encrypted query to the untrusted component, and the untrusted component forwards the query to the enclave. The enclave decrypts the query and asks the untrusted component to load encrypted data from the public memory into the enclave.
The enclave then decrypts the input data, processes the data, and returns the encrypted result to the untrusted component.
The untrusted component forwards the result back to the trusted data owner, who has a decryption key to see the query result. 
During the query processing, the enclave can also send encrypted intermediate results to the untrusted memory and later load them back. This is often needed because enclaves have limited memory. The enclave checks the MACs of the input data and the intermediate results to prevent the cloud server from modifying them.

Unfortunately, the access patterns in the public memory are exposed to the untrusted cloud server. This means an attacker can watch how the enclave reads the encrypted data, writes the encrypted output, and reads/writes the intermediate result.
Data access pattern leakage is sufficient for the attacker to extract secrets and data from many encrypted systems~\cite{theoretical-memory-side-channel, timing-attacks-on-crypto, bernstein2005cache-aes}. Our threat model is the same with ObliDB~\cite{eskandarian2017oblidb}.

%% file: background.tex
\subsection{Differential Obliviousness}\label{sec:models}

\ignore{
To formally define differential obliviousness, we first define differential privacy in Definition~\ref{def:dp}.

\begin{definition}[Differential Privacy]\label{def:dp}

A randomized algorithm $\mathcal{A}$ is $(\epsilon, \delta)-$differential private if for any two \textbf{neighboring databases}
$D_1$, $D_2$ and any subset of possible outputs $O$ of the algorithm:
\begin{align*}
\Pr[\mathcal{A}(D_1) \in O ] \leq e^{\epsilon} \cdot \Pr[\mathcal{A}(D_2) \in O] + \delta .
\end{align*}
Here, the probability is taken over the randomized coin flips of the algorithm $\mathcal{A}$, and the neighboring databases is defined as two databases that only differed by one tuple.
\end{definition}
}

\noindent{\bf Differential obliviousness.}
The notion of differential obliviousness was proposed by Chan et al. \cite{ccms19}. It essentially requires that the memory traces of an algorithm satisfy differential privacy~ \cite{dmns06, dr14}. To provide some background, differential privacy was introduced in the seminal work by Dwork et. al \cite{dmns06, dr14}, which is a framework for 
adding noise to data so that the published result would not harm any individual user's privacy.
Over the years, differential privacy has become the de facto standard for privacy, with growing acceptance in the industry. 
In the differential privacy literature, we typically assume that the data curator is fully trusted, and thus we care about adding noise to the computation result. 
However, in our setting, the data curator (i.e., the cloud provider)
is untrusted. Our goals therefore depart from the standard differential privacy literature. 
Instead of requiring the outputs of the computation to be differentially private, we require that the database system's observable runtime behavior, namely, the access patterns, be differentially private. As mentioned, Chan et al. \cite{ccms19} formulated this notion as differential obliviousness.

With differential obliviousness, the untrusted cloud provider cannot 
extract private information for each individual by observing memory access patterns. A differentially oblivious system is resilient to the attacks mentioned in \autoref{sec:threat}. We formally define differential obliviousness in Definition~\ref{def:do}.

Henceforth, we may view a database as an ordered sequence of records. We say that two databases
$D_1$ and $D_2$ are {\it neighboring}, iff they are of the same length, and moreover, they differ in exactly one record.

\begin{definition}[Differential Obliviousness~\cite{ccms19}]\label{def:do}
An algorithm $\mathcal{A}$ is $(\epsilon, \delta)$-differentially oblivious if for any two {\bf neighboring} {\bf databases}
$D_1$, $D_2$, and any subset of memory access patterns $S$:
\[ \Pr[ \mathcal{M}( \mathcal{A} , D_1) \in S] \leq e^{\epsilon} \cdot \Pr[ \mathcal{M} ( \mathcal{A}, D_2) \in S]  + \delta. \]
\end{definition}
Here, we use $\mathcal{M}(\mathcal{A}, D)$ to denote the distribution of memory access patterns when we apply the algorithm $\mathcal{A}$ on $D$.  The $\epsilon$ parameter is a metric of privacy loss. It also controls the privacy-utility trade-off. The $\delta$ parameter accounts for a negligible probability on which the upper bound $\epsilon$ does not hold.
The memory access pattern is a sequence of memory operations, including the address of each operation and the type of operation (read or write). 
Since the data contents are encrypted, we may assume that the adversary observes only the addresses and types of the operations but not the contents.

It is important to note that in  Definition~\ref{def:do} above, we allow the databases $D_1$ and $D_2$ to contain two types
of records, {\it real} records and {\it filler} records. We allow the filler records so that deleting one entry from the database can be accomplished by replacing the entry with a filler.

Differential obliviousness perfectly captures the threat model enclave-based database systems face: as we discussed in \autoref{sec:threat}, for an enclave-based database system, the data and code execution within the enclave can be considered secure, and the data stored outside the enclave is encrypted but accesses to it leak information. Specifically, in our SGX-based scenario, each memory access observable by the adversary is a page swap event: whenever the SGX enclave wants to swap in or out a new (encrypted) memory page, it needs to contact the untrusted OS for help.

\vspace{3pt}
\noindent{\bf Comparison with full obliviousness.}
It is also instructive to compare the notion of differential obliviousness with the more classical, {\it full obliviousness} notion first proposed by Goldreich~\cite{g87}. We formally define full obliviousness below in Definition~\ref{def:fo}.

\begin{definition}[Full Obliviousness]\label{def:fo}
An algorithm $\mathcal{A}$ is oblivious if for any two databases $D_1$, $D_2$ of the same size and 
any subset of possible memory access patterns $S$:
\begin{align*}
 \Pr[ \mathcal{M}( \mathcal{A}, D_1 ) \in S ] \leq \Pr[ \mathcal{M}( \mathcal{A}, D_2) \in S] + \delta  
 \end{align*}
\end{definition}

Differential obliviousness is a relaxation of full obliviousness in the following senses:
(1) differential obliviousness only requires the memory access patterns over {\it neighboring} databases to be indistinguishable;
(2) the definition of indistinguishability is also relaxed in differential obliviousness, in the sense that we additionally allow a multiplicative $e^\epsilon$ factor when measuring the distance between the two access pattern distributions. 

These relaxations make designing more I/O efficient algorithms possible. For example, in a fully oblivious model, a database system has to add filler tuples to the result until it reaches the worst-case size. In database queries, this worst-case size could be orders of magnitude worse than the average-case size. However, with differential obliviousness, it suffices for the database system to add a small, random number of fillers so that the output size is indistinguishable for two neighboring databases. 

\ignore{
\noindent\textbf{Databases with Filler Tuples.} In differential oblivious setting, one technicality is that the input size could leak sensitive information. For example, if $D'$ is the database that we delete $1$ tuple from $D$, apparently $D$ and $D'$ are neighboring databases. However, the size difference of $D$ and $D'$ itself leaks what is the input data. As a result, in this paper, we assume that all the input, as well as intermediate result and output, contains ploy logarithmic to the overall size filler tuples. If we delete one tuple from the database, we simply
mark this tuple as a filler tuple ($\bot$). 
}

\ignore{
The well-known composition rule (see \autoref{sec:comp}) says that the overall privacy parameter, $\eps$, $\delta$ is additive of the privacy parameters of each operator. 
We refer to the overall privacy parameter $\eps$ and $\delta$ consumed by the entire query as the ``privacy budget''. As mentioned, these budget must be distributed among the operators.}

%% file: operators.tex
\section{Differentially Oblivious Operators}
\label{sec:ops}
In this section, we propose a series of differentially oblivious algorithms 
that implement major relational operators, including 
selection with projection, grouping with aggregation, and foreign key join. 

\paragraph{ Overview of our differentially oblivious operators.}
We propose a differentially oblivious algorithm for selection with projection with optimal cache complexity. The main technique of this algorithm is inspired by a theoretical result on differentially oblivious compaction~\cite{ccms19}: using a differentially private prefix-sum sub-routine
to guide the memory access of filtering (\autoref{sec:sel}).
Next, we propose a differentially oblivious algorithm for grouping with aggregation.
Notably, to develop this algorithm, we propose a novel, practical differentially private distinct count algorithm (\autoref{alg:dpcount}).
This is the \emph{first} practical differentially private streaming algorithm for distinct count with provable approximation guarantees to the best of our knowledge!
We use this algorithm to estimate the number of groups produced and then use a pseudorandom function to partition the input database into smaller partitions such that the groups generated in each partition can fit into the private memory with high probability.  
Last, we present our differentially oblivious foreign key join algorithm 
based on oblivious sort
(\autoref{sec:join}). \vldbrevision{We summarize the notations in Table~\ref{tab:notations}.}

\begin{table}[!ht]
\small
\caption{Notations used in this section}
\vldbrevision{
\begin{tabular}{|c|c|}
\hline
Notation & Description     \\ \hline
$\Pi$    & Projection operator \\ \hline
$\sigma_{\phi}$ & Filter operator with filtering predicate $\phi$         \\ \hline
 $\epsilon$ & The multiplicative factor in differential obliviousness  \\ \hline
  $\delta$ & The additive factor in differential obliviousness   \\ \hline
  $I$ & Input table of size $N$ \\ \hline 
  $P$ & A FIFO buffer in the private memory \\ \hline
  $c$ & Read counter in $I$ \\ \hline
  $\tilde{Y}_c$ & Differentially private prefix sum in first $c$ elements \\ \hline
  $t$ & A single tuple from $I$ \\ \hline
  $L$ & Consisting of grouping attributes and aggregation operators \\ \hline 
  $M$ & Private enclave memory size \\ \hline
  $h$ & Hash function in \textsc{DoGroup}$_h$ \\ \hline
  $\tilde{G}$ & Estimated number of distinct elements in data stream \\ \hline
\end{tabular}
}

\label{tab:notations}
\end{table}

\ignore{

\paragraph{ Distance-preserving requirement.}
One invariant that all our differentially oblivious algorithms carefully maintain is that they are all \emph{distance preserving}. 
This requires that running the algorithm on two neighboring databases must produce outputs that are neighboring too.
Notably, this invariant is not naturally guaranteed in many database systems since most SQL queries do not enforce an explicit order on the output. Distance-preserving property allows the differential obliviousness guarantees to be composed later on, i.e., we can apply a differentially oblivious operator to the outcome of another differentially oblivious operator, and be able to reason about the total privacy budget consumed. We discuss composition in more detail in \autoref{sec:comp}.

}

\subsection{Selection with Projection \texorpdfstring{$(\sigma, \Pi)$}{}}
\label{sec:sel}
A selection operator takes a relation and outputs a subset of the relation according to a filtering predicate. Such an operation is denoted $\sigma_{\phi}(R)$, where $\phi$ is the filtering predicate  and $R$ is the input table.
Intuitively, selection operators act like filtering operations in functional programming languages.
A projection operator transforms one relation into another, possibly with a different schema: it is written ($\Pi_{a_{1},\ldots,a_{n}}(R)$) where $ a_{1},...,a_{n}$ is a set of attribute names. The result of such a projection keeps components of the tuple defined by the set of projected attributes and discards the other attributes. In many database systems, projection is usually inlined in selection. We follow this tradition.
Now, we give the differentially oblivious algorithm for $\sigma_{\phi}(\Pi(R))$, where $\phi$ is the filtering predicate and $R$ is the input table. To better understand our algorithm, we start with a na\"ive non-oblivious algorithm:

\paragraph{Na\"ive non-oblivious algorithm.}
It is clear that a non-private filtering algorithm can achieve
linear time, by reading each input tuple $t$ once and writing it when 
$\phi(t) = \code{TRUE}$. However, this na\"ive algorithm is not 
differentially oblivious. This is because after reading a tuple from input,
whether or not another tuple is written to the output leaks whether the previous tuple from input 
evaluated to \code{TRUE} or \code{FALSE}. Intuitively,
one can visualize the 
memory access pattern of this algorithm using two pointers, a read 
pointer and a write pointer. The attacker observes how fast these two pointers move in each step.

Thus, the main idea of our differentially oblivious filtering 
algorithm, \filteralg, is to obfuscate how fast each pointer advances 
\emph{just enough} to achieve differential obliviousness. 
{\filteralg} is inspired by the theoretical result of differentially 
oblivious stable compaction from \cite{ccms19}. To determine how much noise to add on memory access at each step, 
we query a differentially private oracle 
for computing prefix sum in data streams.

\paragraph{Differentially private prefix-sum.}
\vldbrevision{For a data stream $D$ that
consists of only $0$s and $1$s with length $|D| = n$, $D \in \{0, 1\}^n$}, 
the prefix-sum \vldbrevision{$Y_c$} is the count of 
how many $1$s appear in the first \vldbrevision{$c$} elements \vldbrevision{of data stream $D$}.
Now, suppose we have a $(\eps, \delta)-$differentially private prefix sum algorithm that can answer up to \vldbrevision{$n$} queries, and each answer \vldbrevision{$\tilde{Y}_c \in [Y_c - s, Y_c + s]$} with high probability.
To make the traces of the write pointer differentially private,
we can always move the output pointer to \vldbrevision{$\tilde{Y}_c - s$}, and keep the scanned but not yet output tuples in the private buffer. 
Most importantly, we only need a $2s$ sized buffer in private memory and the algorithm would not encounter errors with high probability. 

We use the binary mechanism of Chan et al.~\cite{css11} 
as our DP prefix-sum oracle. This mechanism essentially builds 
a binary interval tree to store noisy partial sums 
for the optimal approximation-privacy 
trade-off. \vldbrevision{For each $c \in [n]$, the estimated prefix-sum 
$\tilde{Y}_c$} from 
the binary mechanism preserves $\epsilon$-differential privacy
while has $O( \epsilon^{-1} \cdot (\log T) \cdot \sqrt{\log t}
\cdot \log ( 1 / \delta )$ error with at least $1- \delta$ probability
~\cite[Theorem 3.5, 3.6]{css11}.

\begin{algorithm}[t]\caption{\filteralg: Differentially Oblivious Filtering}\label{alg:dofilter}
\begin{algorithmic}[1]
\Procedure{\filteralg}{$I, \Pi, \phi, \epsilon, \delta, s $} \Comment{Theorem~\ref{thm:dofilter}.}
    \State $P \gets \emptyset$ \Comment{a FIFO buffer in private memory of size $2s$} 
    \State $I' \gets \emptyset$ \Comment{output table}
    \State $c \gets 0$ \Comment{current read counter in $I$}
    \State \vldbrevision{$\tilde{Y} \gets \code{DPPrefixSum}(\epsilon, \delta)$   \Comment{$\tilde{Y}_c \in [Y_c - s, Y_c + s]$ with high probability}}
    \While{$c < |I|$}
        \State $T \gets \{I_c, I_{c+1}, \ldots I_{c+s-1}\}$ \Comment{read the next $s$ tuples}
        \State $c \gets c + s $ \Comment{update the read counter}
        \For{$t \in T$}
            \If{$\phi(t) = \code{TRUE}$}
                \State $P \gets P.\code{push}(\Pi(t))$ 
            \EndIf
        \EndFor
        \State Pop $P$ to write $I'$ until $|I'| = \tilde{Y}_c - s$
    \EndWhile
    \State write all tuples from $P$ and filler tuples to 
    $I'$ s.t. $|I'| = \tilde{Y}_N + s$ ($N = |I|$) 
\EndProcedure
\end{algorithmic}
\end{algorithm}

\paragraph{Differentially Oblivious Filtering.}
We present the detailed {\filteralg} in~\autoref{alg:dofilter}. 
\vldbrevision{Let $I$ be the input table of length $|I| = N$.}
Let $s$ be the approximation error (with probability at least $1-\delta$ for each query) of the DP prefix-sum oracle (line 2).
We create a FIFO buffer $P$ in private memory with size $2s$, 
the output table $I'$ outside the private memory, and a counter $c$ to
indicate the number of tuples read so far (line 3-5).
Then we repeat the following until reaching the end of $I$:
we read the next $s$ tuples, and update the counter $c$ (line 7-8). 
For each tuple $t$, we push it to $P$ only if 
the predicate evaluates to \code{TRUE} (line 9-13).
We then pop $P$ to fill the output table $I'$ till it reaches size 
$\tilde{Y}_c - s$ (line 14).
After we reach the end of $I$, we pop all the tuples in $P$ and 
add filler tuples if necessary to append on $I'$ till it reaches
size $\tilde{Y}_N$ where $N = |I|$ (line 16).

\paragraph{Correctness failures to privacy failures.}
\autoref{alg:dofilter} is designed to have correctness failure (i.e. the algorithm \vldbrevision{does not} return the correct result) of probability at most $\delta$. 
This means that \autoref{alg:dofilter} can fail when the DP prefix-sum oracle's estimation
is off by more than $s$. 
When this happens, the size $P$ private memory could either overflow at line 11 or 
underflow (i.e. has nothing to pop) at line 14.
In practice, we do not need to worry about this for two reasons. First, $\delta$ is negligible (usually set to $2^{-20} - 2^{-40}$). Second, in case that users want perfect correctness, we can use the standard technique to convert the correctness failures to privacy failures: Instead of failing, if overflow is about to happen at line 11, we can simply write the overflowed tuple to the output. If underflow happens at line 14, we can write a filler tuple to the output. 
Applying these approaches converts the at most $\delta$ probability correctness error to at most $\delta$ probability privacy error, which is negligible.

\begin{theorem}[Main result for filter]\label{thm:dofilter}
For any  $\epsilon \in (0,1)$, $\delta \in (0,1)$ and input $I$ with $N$ tuples , there is an ($\epsilon, \delta$)-differentially oblivious filtering algorithm ({\filteralg} in Algorithm~\ref{alg:dofilter}) that uses  $O(\log(1/\epsilon) \cdot \log^{1.5}N \cdot \log ( N/\delta))$
private memory and $(N+R)/B$ cache complexity, and its output size is $R + \poly \log(N)$.
\end{theorem}

\begin{proof}
First, we prove \autoref{alg:dofilter} is ($\epsilon, 0$)-differential oblivious if $P$ has infinite capacity:
For two neighboring input $I, I'$, assuming $I, I'$, we only leaks $\tilde{Y}_k$ 
($k=s, 2s, \ldots, n)$. This leakage is bounded by leaking 
all $\tilde{Y}_i$ ($i \in [N]$).
From the DP guarantee provided by the DP prefix sum oracle~\cite{css10}, all writes have at most ($\epsilon, 0$)-DP 
leakage.

Second, we prove \autoref{alg:dofilter} has at most $\delta$ 
probability of privacy failure with 
$O(\log(1/\epsilon) \cdot \log^{1.5} N \cdot \log (N/\delta) )$ private memory. 
Let $s = O(\log(1/\epsilon) \cdot \log^{1.5} N \cdot \log (N/\delta ))$.
We set $P = 2s$. 
Let $Y_c$ denote number of actual filtered tuples generated so far
(at line 14). From 
the DP guarantee provided by the DP prefix sum oracle, 
we know that for each $c$, 
$Y_c - s \leq \tilde{Y}_c \leq Y_c + s$ with $1 - \frac{\delta}{N}$ 
probability. By the union bound, we know that for all rounds of batched read, with at least 
$1 - \delta$ probability, $P$ over-flows or $P$ under-flows:
\begin{align*}
\Pr[Y_c-(\tilde{Y}_c - s) > 2s \lor Y_c < \tilde{Y}_c - s] \geq 1 -\delta    
\end{align*}

Now we can conclude that the failure probability of \autoref{alg:dofilter} is at most $\delta$.
In \autoref{alg:dofilter}, 
the number of tuples we read is $N$, and the number of tuples we write to the output is $\tilde{Y}_N + s$. From the utility-privacy
bound from the DP oracle~\cite{css10}, we know $\tilde{Y}_N + s = R + \poly \log (N)$. Therefore, the output size of the algorithm is $R + \poly \log (N)$.
In addition, all the read and write in Algorithm~\ref{alg:dofilter} are batched, with the batch size $s$. Thus, the cache complexity of \autoref{alg:dofilter} is $(N+R)/B$ as long as $s \geq B$. 
\end{proof}

\begin{remark}
Remark: Note that $\Omega((N+R)/B)$ cache complexity is a trivial lower bound. Thus, our cache complexity is optimal. 

We adapted the technique in \cite{ccms19}, which presents a DO stable compaction algorithm. Stable compaction is a different problem since it keeps all the elements in the input. Also, in~\cite{ccms19}, they don't have a notion of the cache complexity and therefore don't provide any bound for that.
\end{remark}

\subsection{Grouping with Aggregation \texorpdfstring{($\gamma$, $\alpha$)}{}}
\label{sec:groupby}
A grouping operator groups a relation and/or aggregates some columns. It usually denoted $\group{L}{(R)}$ where 
$L$ is the list which consists of two kinds of elements:
grouping attributes, namely attributes of $R$ by which $R$ will be grouped,
and aggregation operators applied to attributes of $R$.
For example, $\group{c_1,c_2,\alpha_1(c_3)}{(R)}$ partitions the 
tuples in $R$ into groups according to attributes $\{c_1, c_2\}$, 
and outputs the aggregation value of $\alpha_3$ on $c_3$ for each group.

We propose a non-oblivious hash based grouping based on randomized 
partitioning on grouping attributes. 
This can be done by applying a pseudorandom function
(PRF) on the grouping attributes $L$. 
One key challenge is to make each partition fit 
into the private memory (size $M$). 
To obtain the correct parameter s
for the randomized partitioning algorithm, we apply 
a preprocessing step. Specifically, we use a randomized 
streaming distinct count algorithm to get the estimated number of groups produced by this query $\tilde{G}$.
As a result, roughly we need $k = \ceil{\tilde{G}/M}$ sequential scans to find all the groups.
To make this algorithm differentially oblivious is yet
another challenge. One first observation is: this algorithm is ``almost'' oblivious if we pad the output 
in each round to $M$, 
except that the number of sequential scans following the preprocessing step leaks information. 
As a result, we need to use a differentially private 
distinct count algorithm. Additionally, we need to bound the failure probability of the randomized partitioning algorithm, such that the size of each partition will not overflow $M$.

Unfortunately, despite few theoretical results  \cite{advstreaming,hkmms20,cgkm20}, 
to the best of our knowledge there is no practical differentially private distinct count streaming algorithm.
To remedy this, we propose a differentially private distinct count algorithm based on the classical distinct count estimator of Bar-Yossef et al. \cite{bjkst02}. The core technique we leverage here is to use properties of uniform order statistics to bound the concentration of both the  approximation error and the sensitivity at the same time. 
\vldbrevision{The general version of this differentially private distinct count algorithm and detailed proofs are in \autoref{sec:dp_distinct}. And we use it to design our differentially oblivious grouping with aggregation algorithm.}

Our differentially private distinct count algorithm (\autoref{alg:dpcount}) first 
creates a priority-queue $P$ of size $t$ in
the private memory (line 2). 
Then, for each element $x_i$ in the stream, our algorithm
applies a PRF $h$ to obtain a hash value of 
the element ($h(x_i) \in [0, 1)$).
Our algorithm uses the priority-queue to keep
the $t$ smallest hash values of the stream (line 4 - 11).
In the end, we pop $P$ to get the $t$-th smallest hash value $v$ (line 13).
Finally, we output the estimated value of distinct count in line 14. The unbiased estimation should be $t/v$, as stated in \cite{bjkst02}. Here, since the estimated value is used to calculate the number of partitions needed, we can only over estimate. We need to add proper noise to make the algorithm differentially private as well. As a result, the algorithm outputs the noisy count as shown in line 14. \vldbrevision{We will use this algorithm with approximation parameter $\eta = 0.1$ to design our differentially oblivious grouping algorithm.}

\begin{algorithm}[t]\caption{DoGroup$_h$: Differentially Oblivious Grouping}\label{alg:dogrouph}
\begin{algorithmic}[1]
\Procedure{DOGroup$_h$}{$I, L, \epsilon, \delta$} \Comment{Theorem~\ref{thm:dogrouph}}
    \State \vldbrevision{$\tilde{G} \gets \code{DPDistinctCount}(I, \epsilon, 0.1, \delta/2)$ \Comment{\autoref{alg:dpcount} with approximation parameter $\eta = 0.1$}}
    \State $k \gets \ceil{\tilde{G}/0.9M}$ \Comment{$M$: size of the private memory}
    \State Verify that $\sqrt{ 0.5 \tilde{G} \log (2k/\delta) } \leq 0.1M $
    \State $R \gets \emptyset$ \Comment{output table}
    \For{$i \in {0, \ldots, k-1}$}
        \State $\mathcal{H} \gets \emptyset$ \Comment{Hash table for grouping}
        \For{$t \in I$}
            \If{$h(t.L) \in [i/k, (i+1)/k)$} 
            \State \Comment{$h: [m \times \ldots \times m ] \leftarrow [0, 1) $, is a PRF}
                \If{$\mathcal{H}.\code{hasKey}(t.L)$}
                   \State update $\mathcal{H}(t.L)$'s aggregate values using $t$
                \Else
                   \State $\mathcal{H}(t.L) \gets t$
                \EndIf
            \EndIf
        \EndFor
        \State write all tuples in $\mathcal{H}$ and filler tuples to $R$, s.t. $R$'s size increased by $M$ 
    \EndFor
\EndProcedure
\end{algorithmic}
\end{algorithm}

In \autoref{alg:dogrouph}, we present our differentially oblivious grouping algorithm.
\vldbrevision{Let $I$ be the input table. Let $L$ be the list which consists of two kinds of elements:
grouping attributes and aggregation operators.}
The algorithm first computes the
1.1-approximate differentially private distinct count $\tilde{G}$
(line 2), and use $\tilde{G}$ to calculate the number of partitions
 $k = \ceil{\tilde{G}/0.9M}$ (line 3). We set this $k$ value to overestimate the number of iterations we need in order to obtain a negligible failure probability of out-of-enclave memory error among k iterations.
 To ensure the size of each partition is less than $M$ with at least $1 - \delta$ probability, we verify 
 that $\sqrt{ 0.5 \tilde{G} \log (2k/\delta) } \leq 0.1M$ (line 4).
 Next, the algorithm sequentially scans $I$ a total of $k$ times. In $i$-th scan, the algorithm creates an empty hash table $\mathcal{H}$ (line 7). Then, for each tuple $t$, the algorithm applies a PRF $h$ on the list of grouping attributes of $t$. Here $h(t.L)$. $h(t.L)$ falling into $[i/k, (i+1)/k)$ means this group is within the partitioned groups of current sequential scan. In this case, the algorithm either update the aggregate values if $\mathcal{H}$ already contains $t$, or create a new entry for $t$ in $\mathcal{H}$ (line 9 - 16). In the end of each sequential scan, we
 output all groups in $\mathcal{H}$ and filler tuples so that 
 size $M$ is written to the output (line 18).

\begin{theorem}[Main result for grouping]\label{thm:dogrouph}
For any  $\epsilon \in (0,1)$, $\delta \in (0,1)$ and input $I$ with size $N$ and $O(\eps^{-1}\log^2(1/\delta))$ private memory $M$, there is an ($\epsilon, \delta$)-differentially oblivious and distance preserving grouping algorithm ({\grouphalg} in Algorithm~\ref{alg:dogrouph}) that uses  $M$
private memory, has $N/B+11NR/9MB$ cache complexity and its output size is $\frac{11}{9}R$.
\end{theorem}

\begin{proof}

\vldbrevision{Because at the end of each scan, we write to the output table until its size is increased by $M$ and the enclave private memory size $M$ is public, }
the only information that \autoref{alg:dogrouph} leaks is
$k$, the number of sequential scans of $I$.
This only leaks $\tilde{G}$ though. Moreover, $\tilde{G}$ is $(\eps, \delta/2)$-differentially private. \vldbrevision{And the differentially private distinct count algorithm is oblivious.}
Thus, \autoref{alg:dogrouph} is $(\eps, 0)$-differentially oblivious if we ignore the failure case.

Next, we prove that the failure probability of \autoref{alg:dogrouph} is at most $\delta$. 
Since $\sqrt{ 0.5 \tilde{G} \log (2k/\delta) } \leq 0.1M $ (line 4) and the expected number of group generated in each sequential scan is $0.9M$, let $G_i$
be the number of groups $i$-th sequential scan generated. From the properties of binomial distribution (detailed lemmas can be found in \autoref{sec:rand_partition}),
we have $\Pr[G_i \leq M] \leq \delta/2k$.
Applying a union bound over $k$ scans, we can bound the failure probability of the randomized partitioning by at most $\delta/2$. Applying union bound again with the $(\eps, \delta)$-differentially private $\tilde{G}$, we can bound the failure probability of \autoref{alg:dogrouph} to at most $\delta$. 

Finally, \autoref{alg:dogrouph} requires a sequential scan of input in preprocessing, whose cache complexity is $N/B$.
In the following steps, the cache complexity is $k(M + N) /B$, where $k \leq 1.1 R/ 0.9M$, which is no more than 
$\frac{11R(M+N)}{9MB}$. Here $M$ is absorbed by $N$. Thus, the overall cache complexity of \autoref{alg:dogrouph} is $N/B + 11NR/9MB$. Because we need to write $M$ groups to the output for $k$ passes, the output size is $\frac{11}{9}R$.

\end{proof}

\subsection{Foreign Key Join \texorpdfstring{($\bowtie$)}{}}
\label{sec:join}
Foreign key join is the most widely used join operator in data analytics. We write $R \bowtie S$\footnote{$\bowtie$ is normally used for natural join, we abuse the notion here.} to represent 
a join on the primary key and foreign key pairs of the relations $R$ and $S$

A typical oblivious join algorithm first pads the tuples from two joined tables to the same size, and add a ``mark'' column to every tuple to mark which table is this tuple from. Then, it performs an oblivious sort on the concatenation of both joined tables.
This oblivious sort routes the tuples to be joined from both tables
to the same group. 
Next, the algorithm makes a sequential scan of the sorted table to generate the result table. This can be done obliviously since for each tuple read from the primary key table, the algorithm will output a filler tuple instead. 
Last, the algorithm uses another oblivious sort to remove all the filler tuples.
This algorithm is first implemented in Opaque~\cite{zheng2017opaque} and then followed by ObliDB~\cite{eskandarian2017oblidb}.

We develop \textsc{DoJoin}, a differentially oblivious foreign key join algorithm. In \textsc{DoJoin}, the neighboring databases $D_1$ and $D_2$ should contain the same primary key table and their foreign-key tables have the same length but differ in one record.
\textsc{DoJoin} improves the standard oblivious foreign key join in three aspects. 
First, 
we use the more efficient bucket oblivious sort~\cite{acnprs20} replacing the 
bitonic sort used in ObliDB. Compared with bitonic sort, bucket oblivious sort 
has better asymptotic complexity ($O(n \log n)$ compared with $O(n \log^2 n)$)
and still relatively small constant $6$.
Second, our algorithm only sorts the input once, removing filler
tuples is done by our differential oblivious filtering algorithm (\autoref{alg:dofilter}, \autoref{sec:sel}).
Lastly, our algorithm pads filler tuples in the output to the size of differential 
obliviousness requirement, rather than to the worse case size,
which could be much smaller in practice.

\begin{algorithm}[!ht]\caption{\textsc{DoJoin}: DO Foreign Key Join}\label{alg:dojoin}
\begin{algorithmic}[1]
\Procedure{\textsc{DoJoin}}{$R, S, k_R, k_S, \epsilon, s$} \Comment{Theorem~\ref{thm:dojoin}}
\State \Comment{$k_R$ is the PK of $R$, $k_S$ is a FK in $S$ referring $k_R$ }
    \State pad the size of each row of $R$ and $S$ to the 
    greater row size of $R$ and $S$
    \State $R' \gets \rho(\texttt{mark}(R, \text{`r'}), k_R \rightarrow k)$
    \State $S' \gets \rho(\texttt{mark}(S, \text{`s'}), k_S \rightarrow k)$
    \State $I \gets \code{BucketObliviousSort}(R' || S', k || mark)$

    \State $t \gets \bot$ \Comment{Current tuple from $R$ to be joined}
    \State $I' \gets \emptyset$
    \For{$x_i \in I$}
        \If{$x_i.mark =$ `r'}
        \State $t \gets x_i $
        \State $I' \gets I' || \bot$ \Comment{$\bot$ means filler tuple}
        \Else \Comment{$x_i.mark =$ `s'}
        \State $I' \gets I' || (x_i \cup t)$
        \EndIf
    \EndFor
    \State $T \gets \code{DoFilter}(I', \code{ID}, \lambda t. t \neq \bot, \epsilon, \delta)$ \Comment{remove filler tuples, Algorithm~\ref{alg:dofilter}.}
    \State \Return
\EndProcedure
\end{algorithmic}
\end{algorithm}

We present \textsc{DoJoin} in \autoref{alg:dojoin}.
\vldbrevision{Let $R$ and $S$ be the primary key and foreign key tables. $k_R$ is the PK of $R$, $k_S$ is a FK in $S$ referring $k_R$.}
\textsc{DoJoin} first pads tuples from $R$ and $S$ to the same size and adds an additional ``mark column'' to each tuple to mark which relation it comes from.
This result in $R'$ and $S'$ (line 3 - 5).  Next, \textsc{DoJoin} concatenates $R'$ and $S'$ ($R' || S'$)
and then sort the result first by the key column and then by the mark column. For tuples with the same key, the tuple from $R'$ will always be read first (if it exists). 
Now, \textsc{DoJoin} sequentially scans the sorted table: 
if a tuple from $R'$ is scanned, assign it to the working tuple, 
and output a filler tuple to the output (line 9 - 11); if a tuple from $S'$ is 
scanned, we join it with the working tuple and write the joined tuple to the output (line 13).
Lastly, we call \textsc{DoFilter} to remove the filler tuples.

\begin{theorem}[Main result for join]\label{thm:dojoin}
For any  $\epsilon \in (0,1)$, $\delta \in (0,1)$, input $I$ with size $N$,
private memory of size $M$ and result size $R$, there is an ($\epsilon, \delta$)-differentially oblivious and distance preserving foreign key join algorithm (DoJoin in \autoref{alg:dojoin}) that uses  $O(\log(1/\epsilon) \cdot \log^{1.5}N \cdot \log ( N/\delta))$
private memory and has $6 (N/B) \log(N/B) + (N+R)/B$ cache complexity, and its output size is $R + \poly \log(N)$.
\end{theorem}

\begin{proof}
\textsc{DoJoin} is $(\eps, \delta)$-differentially oblivious follows that \textsc{BucketObliviousSort} is fully oblivious and 
\textsc{DoFilter} is $(\eps, \delta)$-differentially oblivious
with $O(\log(1/\epsilon) \cdot \log^{1.5}N \cdot \log ( N/\delta))$ private memory. 

\textsc{DoJoin} requires sorting $R' || S'$ obliviously once. It uses \textsc{BucketOblivousSort}, which has cache complexity of
$6(N/B) \log{N/B}$. Additionally, the \textsc{DoJoin} algorithm uses the \textsc{DoFilter} algorithm which has a cache complexity $(N+R)/B$. Thus, the cache complexity of \textsc{DoJoin} is $6 (N/B) \log(N/B) + (N+R)/B$. From the output size of \textsc{DoFilter} algorithm, we know the output size of \textsc{DoJoin} is $R + \poly \log(N)$.
\end{proof}

%% file: dp_distinct.tex
\section{Differentially Private Distinct Count} \label{sec:dp_distinct}
In this section, we describe a differentially private distinct count algorithm based on \cite{bjkst02}. We first prove the main technical lemmas 
about order statistics properties of random sampling in \autoref{sec:tech_lemma}. Next, we define $(\eps, \delta)$-sensitivity and 
introduce the Laplacian mechanism for $(\eps, \delta)$-sensitivity in \autoref{sec:newsensitivity}. This follows by our analysis of \cite{bjkst02}: its $(\eps, \delta)$-sensitivity and its 
approximation ratio concentration.
Last, we develop a differentially private distinct count algorithm (Algorithm~\ref{alg:dpcount}) based
on \cite{bjkst02}, \vldbrevision{which is used to implement \textsc{DoGroup}$_h$ (Algorithm~\ref{alg:dogrouph}) with approximation parameter $\eta = 0.1$}. 

\subsection{Order Statistics Properties of Random Sampling}\label{sec:tech_lemma}
For any integer $n$, we use $[n]$ to denote the set $\{1,2,\cdots,n\}$.
Let $x_1,x_2,\dots,x_n \sim [0,1]$ be independently and uniformly sampled, and let $x_{(1)} \leq x_{(2)} \leq \dots \leq x_{(n)}$ be the order statistics of the samples $\{x_i\}_{i=1}^n$. For simplicity, we write $y_i = x_{(i)}$ to align with the notation above, so that $y_i$ is the $i$-th smallest value in $\{x_i\}_{i=1}^n$. Fix any $1 \leq t \leq n/2$. Our goal is to prove that $| \frac{1}{y_t} - \frac{1}{y_{t+1}}| \leq O(\frac{n}{t^2})$ with large constant probability. We begin with the following simple claim which lower bounds $y_t$. We note that the bound improves for larger $t$, so one can use whichever of the two bounds is better for a given value of $t$. We delay the proof of Claim~\ref{claim:1} to Section~\ref{sec:tech_lemma:app}.

\begin{claim}\label{claim:1}
Let $t \in [n]$, and fix any $0 < \delta < 1/2$. Then we have the following two bounds:

\begin{enumerate}
    \item $\Pr[ y_t > \delta \frac{t}{n} ] \geq 1 - \delta$. 
    \item $\Pr[ y_t >  \frac{t}{2n} ] \geq 1 - \exp(-t/6)$. 
\end{enumerate}

\end{claim}

We now must lower bound $y_{t+1}$, which we do in the following claim. We delay the proof to Section~\ref{sec:tech_lemma:app}.

\begin{claim}\label{claim:2}
Fix any $4 <\alpha <n/2$, and $1 \leq t \leq n/2$. Then we have
\begin{align*}
  \Pr[  y_{t+1} < y_t + \alpha/n ] \geq 1- \exp( - \alpha / 4 ) .
\end{align*}
\end{claim}

We then want to bound $ | \frac{1}{y_t} - \frac{1}{y_{t+1}}|$ in the following lemma and delay the proof to Section~\ref{sec:tech_lemma:app}.
\begin{lemma}\label{lem:main}
Fix any $0 < \beta \leq 1/2$, $1 \leq t \leq n/2$, and $\alpha$ such that $4 < \alpha < \beta t/2$. Then we have the following two bounds:

\begin{enumerate}
    \item $\Pr[ | \frac{1}{y_t} - \frac{1}{y_{t+1}}| < \frac{\alpha}{\beta^2} \frac{n}{t^2} ] \geq 1-\beta - \exp( - \alpha / 4 ) $.
    \item $ \Pr[ | \frac{1}{y_t} - \frac{1}{y_{t+1}}| < 4 \alpha \frac{n}{t^2} ] \geq 1- \exp( - t / 6 ) - \exp( - \alpha / 4 ) $.
\end{enumerate}
\end{lemma}

\begin{remark}

Notice that the above lemma is only useful when $t$ is larger than some constant, otherwise the bounds $4 < \alpha < \delta t/2$ for  $0 < \delta < 1/2$ will not be possible. Note that if we wanted bounds on $| \frac{1}{y_t} - \frac{1}{y_{t+1}}|$ for $t$ smaller than some constant, such as $t=1,2,$ ect. then one can simply bound $| \frac{1}{y_t} - \frac{1}{y_{t+1}}| < \frac{1}{y_t}$ and apply the results of Claim \ref{claim:1}, which will be tight up to a (small) constant. 
\end{remark}

\subsection{\texorpdfstring{$(\eps,\delta)$}{}-Sensitivity}\label{sec:newsensitivity}

In what follows, let $\mathcal{X}$ be the set of databases, and say that two databases $X,X' \in \mathcal{X}$ are neighbors if $\|X-X'\|_1 \leq 1$.
\begin{definition}[$\ell$-sensitivity~\cite{dr14}]\label{def:puresensitive}
Let $f: \mathcal{X} \to \R$ be a function. We say that $f$ is $\ell$-sensitivity if for every two neighboring databases $X,X' \in \mathcal{X}$, we have $|f(X) - f(X')| \leq \ell$.
\end{definition}
\begin{theorem}[The Laplace Mechanism \cite{dmns06}]
Let $f: \mathcal{X} \to \R$ be a function that is $\ell$-sensitive. Then the algorithm $A$ that on input $X$ outputs $A(X) = f(X) +\text{Lap}(0,\ell/\eps)$ preserves $(\eps,0)$-differential privacy.
\end{theorem}

In other words, we have $\Pr[ A(X) \in S ] = (1 \pm \eps) \Pr[ A(X') \in S ]$ for any subset $S$ of outputs and neighboring data-sets $X,X' \in \mathcal{X}$. 
\vldbrevision{We now introduce a small generalization of pure sensitivity (Definition \ref{def:puresensitive}), that allows the algorithm to \textit{not} be sensitive with a very small probability $\delta$. The difference between $\ell$-sensitivity and $(\ell,\delta)$-sensitivity is precisely analogous to the difference between $\eps$-differential privacy and $(\eps,\delta)$-differential privacy, where in the latter we only require the guarantee to hold on a $1-\delta$ fraction of the probability space. Thus, to achieve $(\eps,\delta)$-differential privacy (as is our goal), one only needs the weaker $(\ell,\delta)$ sensitivity bounds.}

\begin{definition}[$(\ell,\delta)$-sensitive]\label{def:generalsensitive}
Fix a randomized algorithm $\mathcal{A}: \mathcal{X} \times R \to \R$ which takes a database $X \in \mathcal{X}$ and a random string $r \in R$, where $R = \{0,1\}^m$ and $m$ is the number of random bits used. We say that $\mathcal{A}$ is $(\ell,\delta)$-sensitive if for every $X \in \mathcal{X}$ there is a subset $R_X \subset R$ with $|R_X| > (1-\delta)|R|$ such that for any neighboring datasets $X,X' \in \mathcal{X}$ and any $r \in R_X$ we have $|\mathcal{A}(X,r) - \mathcal{A}(X',r)| \leq \ell$
\end{definition}
Notice that our algorithm for count-distinct is $(O(\alpha\frac{n}{t}),O(e^{-t} + e^{-\alpha}))$-sensitive, following from the technical lemmas proved above. We now show that this property is enough to satisfy $(\eps,\delta$)-differential privacy after using the Laplacian mechanism.

\begin{lemma}\label{lem:generaldp}
Fix a randomized algorithm $\mathcal{A}: \mathcal{X} \times R \to \R$ that is $(\ell,\delta)$-sensitive. Then consider the randomized laplace mechanism $\overline{\mathcal{A}}$ which on input $X$ outputs $\mathcal{A}(X,r) + \text{Lap}(0,\ell/\eps)$ where $r \sim R$ is uniformly random string. Then the algorithm $\overline{\mathcal{A}}$ is $(\eps,2(1+e^\eps) \delta )$-differentially private.
\end{lemma}
The proof is delayed to Section~\ref{sec:newsensitivity:app}.

\subsection{Analysis of Distinct Count}

In this section, we thoroughly analyze the properties of Distinct Count~\cite{bjkst02}. We first describe the algorithm in Algorithm~\ref{alg:count}. Then we prove its $(\ell, \delta)$-sensitivity in Lemma~\ref{lem:dc-sensitivity} and a tighter $(\eps, \delta)$-approximation result in Lemma~\ref{lem:approximation} (compared with the approximation result in~\cite{bjkst02}).

\begin{algorithm}[!ht]\caption{Distinct Count~\cite{bjkst02}}\label{alg:count}
\begin{algorithmic}[1]
\Procedure{\textsc{DistinctCount}}{$I, t$} \Comment{Lemma~\ref{lem:dc-sensitivity}}
    \State $d \gets \emptyset$ \Comment{$d$ is a priority-queue of size $t$}
    \For{$x_i \in I$}
        \State $y \gets h(x_i) $ \Comment{$h$: $[m] \rightarrow [0,1]$, is a PRF}
        \If{$|d| < t$} 
        \State $d.\code{push}(y)$
        \ElsIf{$y < d.\code{top}() \; \land \; y \notin d$} 
        \State $d.\code{pop}()$
        \State $d.\code{push}(y)$
        \EndIf
    \EndFor
    \State $v \gets d.\code{top}()$
    \State \Return $t/v$
\EndProcedure
\end{algorithmic}
\end{algorithm}

\noindent\textbf{Sensitivity of distinct count.} By analyzing the Distinct Count Algorithm~\ref{alg:count}, we show that it is $(20\log(4/\delta) \frac{n}{t}, \delta)$-sensitive in Lemma~\ref{lem:dc-sensitivity}. We will use its sensitivity to design our differential private distinct count algorithm~\ref{alg:dpcount}.

\begin{lemma}[Sensitivity of DistinctCount]\label{lem:dc-sensitivity}
Assume $r \in R$ is the source of randomness 
of the PRF in DistinctCount (Algorithm~\ref{alg:count}), 
where $R \in \{0, 1\}^m$, $n$ is the number of 
distinct element of the input,
for any $16 < t < n/2$, DistinctCount is $(20\log(4/\delta) \frac{n}{t}, \delta)$-sensitive.
\end{lemma}
The proof is delayed to Section~\ref{sec:analysis_distinct_count:app}.

\noindent\textbf{Lemma for approximation guarantees.} Then we show the approximation guarantees for the Distinct Count with high probability in Lemma~\ref{lem:approximation}.

\begin{lemma}\label{lem:approximation}
Let $x_1,x_2,\dots,x_n \sim [0,1]$ be uniform random variables, and let $y_1,y_2,\dots,y_n$ be their order statistics; namely, $y_i$ is the $i$-th smallest value in $\{x_j\}_{j=1}^n$. Fix $\eta \in (0,1/2),\delta \in (0, 1/2)$. Then if $t > 3(1+\eta) \eta^{-2} \log(2/\delta)$, with probability $1-\delta$ we have 
\begin{align*}
    (1-\eta) \cdot n \leq \frac{t}{y_t} \leq (1+\eta) \cdot n .
\end{align*}     
\end{lemma}

The proof is delayed to Section~\ref{sec:analysis_distinct_count:app}.

\subsection{Differentially Private Distinct Count}

\begin{algorithm}[t]\caption{\textsc{DPDistinctCount}: Differentially Private Distinct Count}\label{alg:dpcount}
\begin{algorithmic}[1]
\Procedure{\textsc{DPDistinctCount}}{$I, \eps, \eta, \delta$} \Comment{Theorem~\ref{thm:main}}
    \State $\textsc{pQueue} \gets \emptyset$ \Comment{$\textsc{pQueue}$ is a priority-queue of size $t$}     \Comment{ $t \geq \max\big(3(1+ {\eta}/{4})({\eta}/{4})^{-2}\log({6}/{\delta}),\; 20 \eps^{-1} ({\eta}/{4})^{-1} \cdot \log({24(1+e^{-\eps})}/{\delta}) \cdot \log({3}/{\delta}) )$}

    \For{$x_i \in I$}
        \State $y \gets h(x_i) $ \Comment{$h$: $[m] \rightarrow [0,1]$, is a PRF}
        \If{$|\textsc{pQueue}| < t$} 
        \State $\textsc{pQueue}.\code{push}(y)$
        \ElsIf{$y < \textsc{pQueue}.\code{top}() \; \land y \notin \textsc{pQueue}$} 
        \State $\textsc{pQueue}.\code{pop}()$
        \State $\textsc{pQueue}.\code{push}(y)$
        \EndIf
    \EndFor
    \State $v \gets \textsc{pQueue}.\code{top}()$
    \State $\vldbrevision{\tilde{G}} \leftarrow (1+\frac{3}{4} \eta ) \frac{t}{v} + \text{Lap}(20 \eps^{-1}\frac{n}{t}\log(24(1+e^{-\eps})/\delta))$
    \State \Return $\vldbrevision{\tilde{G}}$
\EndProcedure
\end{algorithmic}
\end{algorithm}

We present our main result for differentially private distinct count algorithm below:
\begin{theorem}[main result]\label{thm:main}
For any $0 < \eps < 1$, $0 < \eta < 1/2$, $0 < \delta < 1/2$, there is an distinct count algorithm (Algorithm~\ref{alg:dpcount}) such that:
\begin{enumerate}
    \item The algorithm is $(\eps, \delta)$-differentially private.
    \item  With probability at least $1 - \delta$, the estimated distinct count $\tilde{A}$ satisfies: 
    \begin{align*}
        n \leq \vldbrevision{\tilde{G}} \leq (1+\eta) \cdot n,
    \end{align*}
    where $n$ is the number of distinct elements in the data stream.
\end{enumerate}
The space used by the distinct count algorithm is \begin{align*}
O\Big( ( \eta^{-2} + \eps^{-1} \eta^{-1}  \log(1/\delta)) \cdot \log(1/\delta) \cdot \log n \Big)
\end{align*}
bits.
\end{theorem}

\vldbrevision{The proof is delayed to Section~\ref{sec:dpdc:app}.}

\begin{claim}\label{cla:t-max}
For any $0 < \delta \leq 10^{-3}$, $0.1  \leq \eta < 1$ 
and $0 < \eps < 1$, then we have 
\begin{align*}
    & ~ 3( 1 + \eta / 4 ) \cdot ( \eta / 4 )^{-2} \cdot \log(6/\delta) \\
     & ~ \leq 25  \eps^{-1} ( \eta/ 4 )^{-1} \cdot \log(24(1+e^{-\eps})/\delta) \cdot \log(3/\delta) .
\end{align*}     
\end{claim}
\vldbrevision{The proof is delayed to Section~\ref{sec:dpdc:app}.}

\begin{lemma}\label{lem:approxdp}
For any $0 < \eps < 1$, $0 < \delta \leq 10^{-3}$, there is an distinct count algorithm (Algorithm~\ref{alg:dpcount}) such that:
\begin{enumerate}
    \item The algorithm is $(\eps, \delta)$-differentially private.
    \item  With probability at least $1 - \delta$, the estimated distinct count $\vldbrevision{\tilde{G}}$ satisfies: 
    \begin{align*}
        n \leq \vldbrevision{\tilde{G}} \leq 1.1 n,
    \end{align*}
    where $n$ is the number of distinct elements in the data stream.
\end{enumerate}
The space used by the distinct count algorithm is \begin{align*}
O\Big( ( 100 + 10 \eps^{-1} \log(1/\delta)) \cdot \log(1/\delta) \cdot \log n \Big)
\end{align*}
bits.
\end{lemma}
\vldbrevision{The proof is delayed to Section~\ref{sec:dpdc:app}.}

%% file: exp.tex
\section{Measuring Empirical Speedup}
\label{sec:eval}
We evaluate our DO operators on Big Data Benchmark~\cite{bdb}.
We compare our performance with 
ObliDB~\cite{eskandarian2017oblidb}, and Spark SQL~\cite{sparksql}. 
We run our experiments on a machine with Intel Core-i7 9700 (8 cores @ 3.00GHz, 12\,MB cache). The machine has SGX hardware and 64GB DDR4 RAM, and it runs Ubuntu 18.04 with SGX Driver version 2.6, SGX PSW version 2.9, and SGX SDK version 2.9. \vldbrevision{We set the SGX max heap size as 224 MB and EPC page swapping will be triggered during processing large tables.
}
We fill  data from the Big Data Benchmark~\cite{bdb}.
We evaluate the performance under three tiers of input table sizes: For filter operator benchmark, \textit{Rankings} table contains \vldbrevision{small (100K), medium (1M), large (10M)} rows \vldbrevision{and each \textit{Rankings} row is $308$ bytes}. For groupby operator benchmark, \textit{UserVisits} table contains \vldbrevision{small (300K), medium (3M) and large (30M) rows} \vldbrevision{ and each \textit{UserVisits} row is $529$ bytes}. For foreign key join benchmark, \textit{Rankings} and \textit{UserVisits} contain \vldbrevision{small (100K, 300K), medium (300K, 900K), large (1M, 3M) rows.}
All codebases are compiled and run under SGX prerelease and hardware mode. We do not compare our operators with those in Opaque~\cite{zheng2017opaque}. This is because Opaque's open-sourced version does not pad the result of an operator to the worse-case length, which means Opaque's open-sourced version does not satisfy the notion of full obliviousness.

\begin{figure*}[!h]
\centering
\subfloat[Selection with projection-100K ]{\includegraphics[width = 0.36 \textwidth]{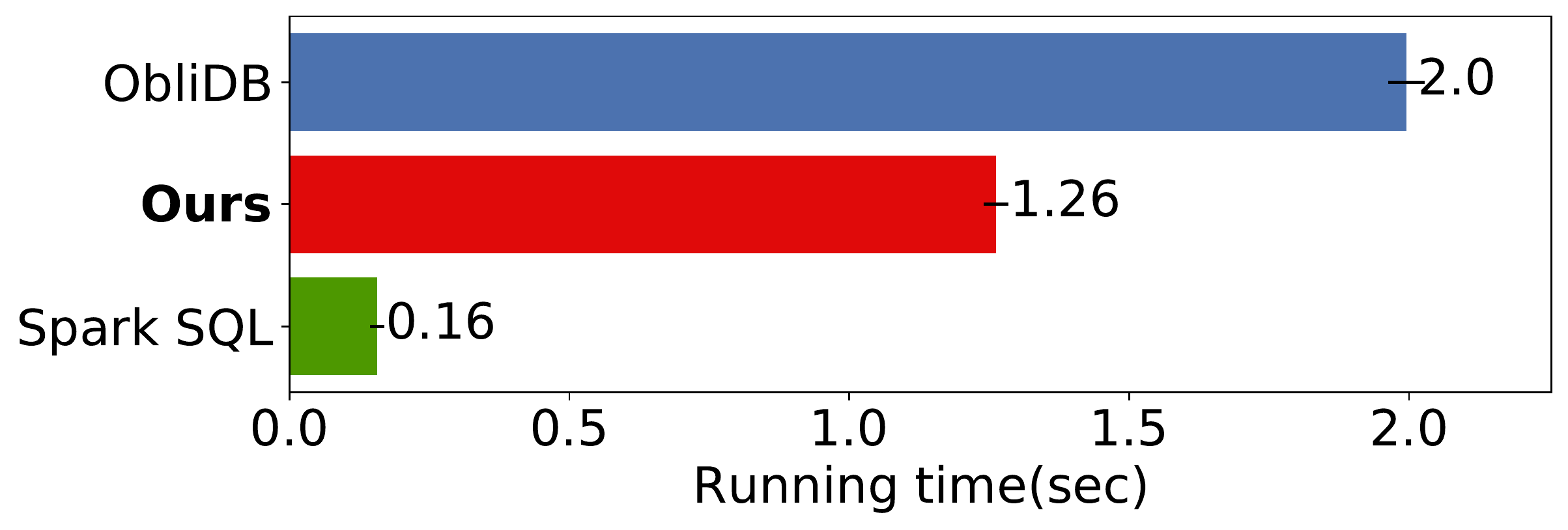}}
\subfloat[Selection with projection-1M ]{\includegraphics[width = 0.31 \textwidth]{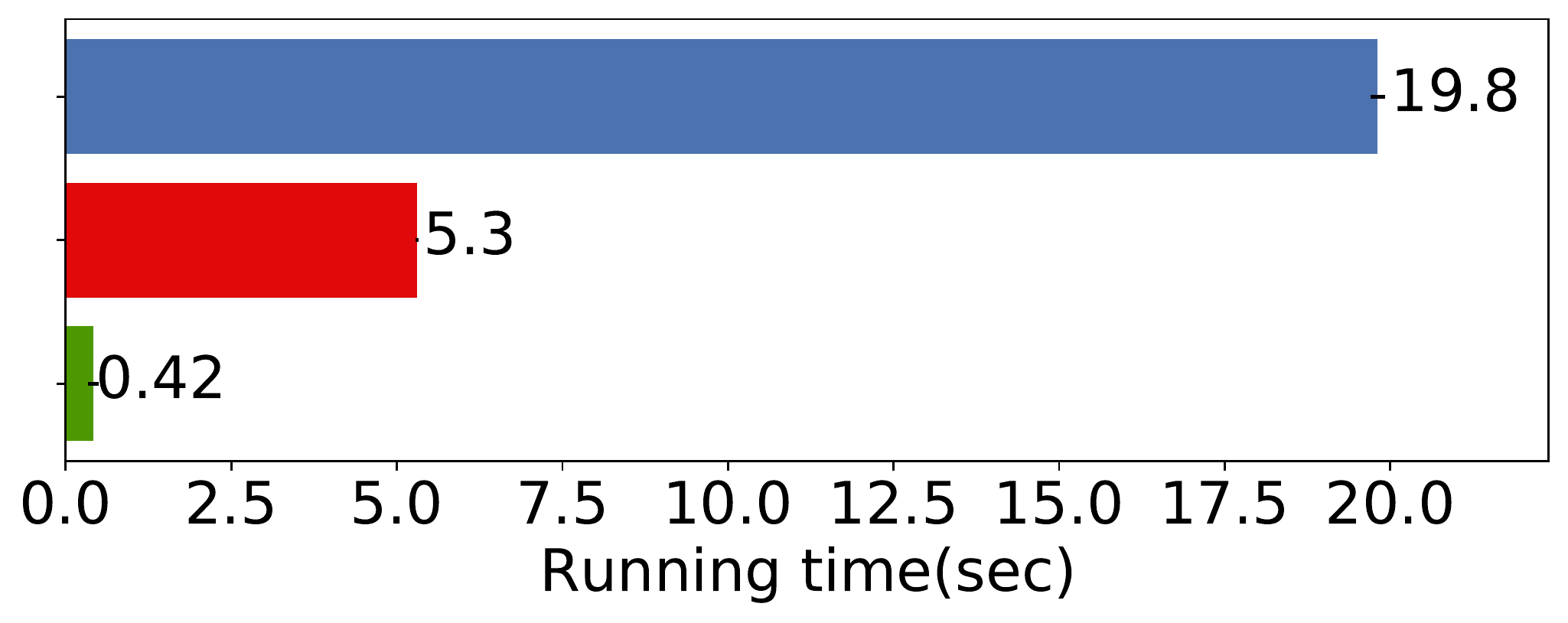}}
\subfloat[Selection with projection-10M]{\includegraphics[width = 0.31 \textwidth]{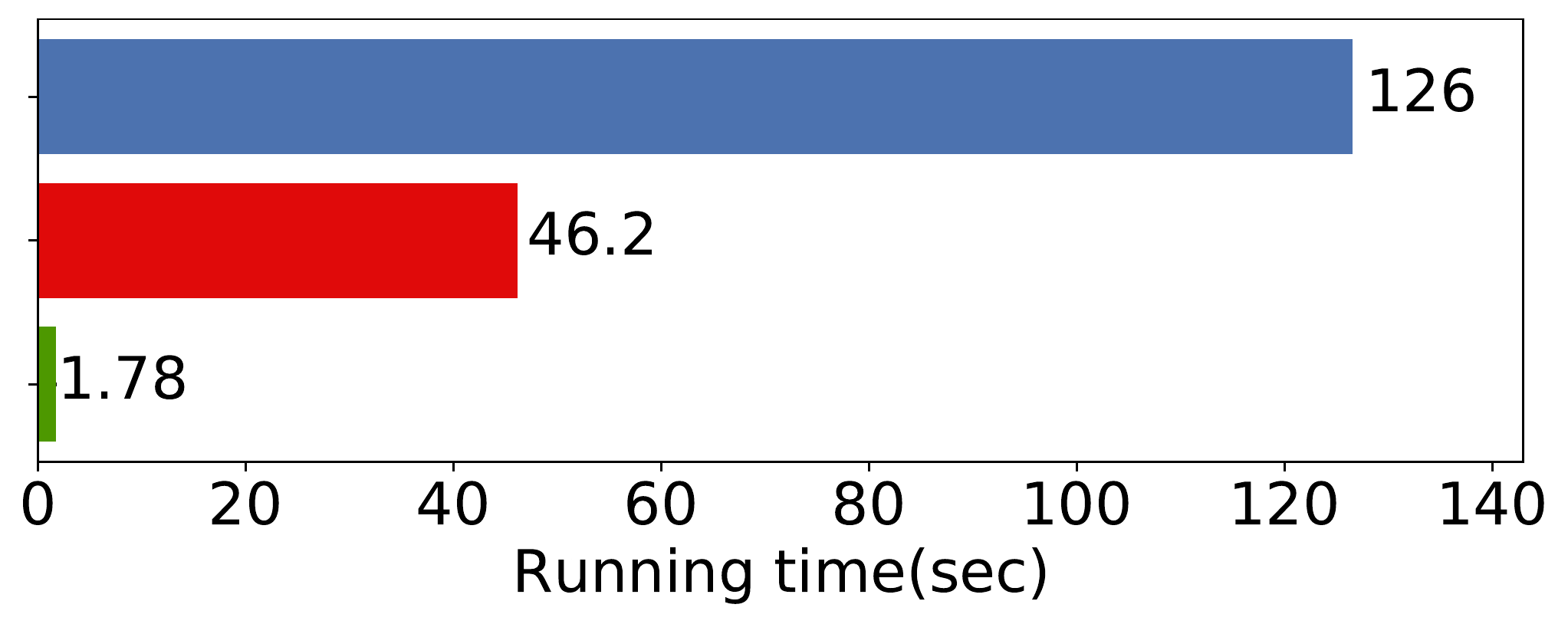}}
\caption{{ Selection with projection operator performance under different input sizes. Error bars show the standard deviations.}}
\label{fig:BDB1}

\end{figure*}
\begin{figure*}[!h]
\centering
\subfloat[Grouping with aggregation-300K ]{\includegraphics[width = 0.37\textwidth]{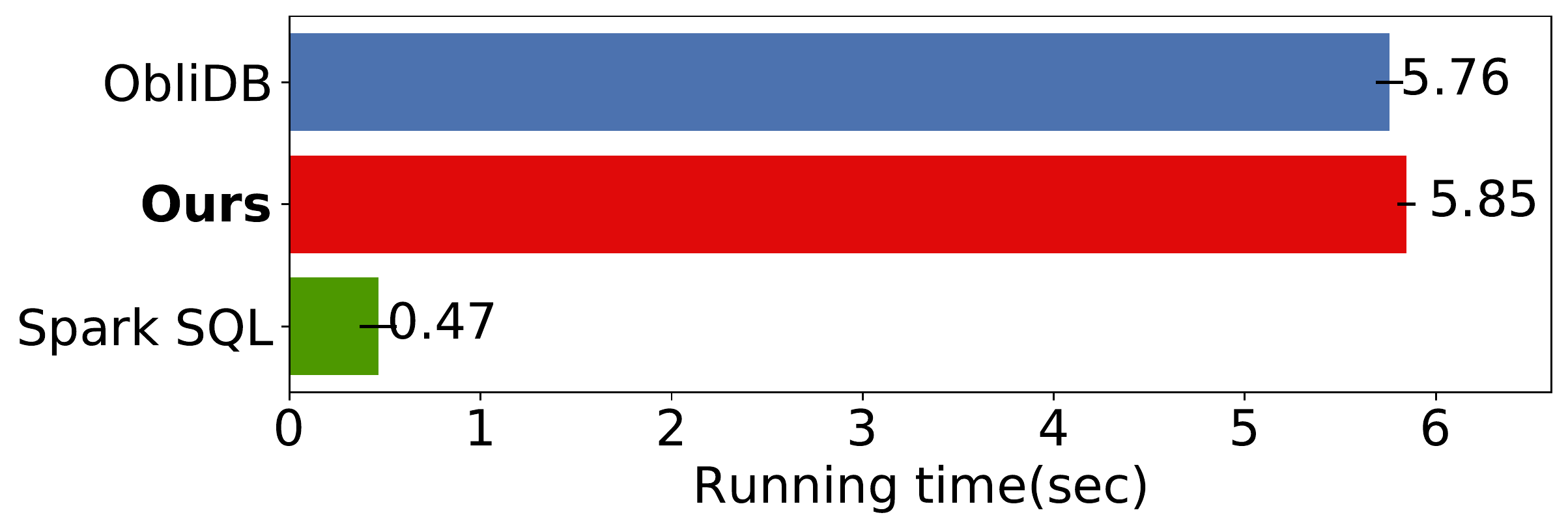}}
\subfloat[Grouping with aggregation-3M ]{\includegraphics[width = 0.31\textwidth]{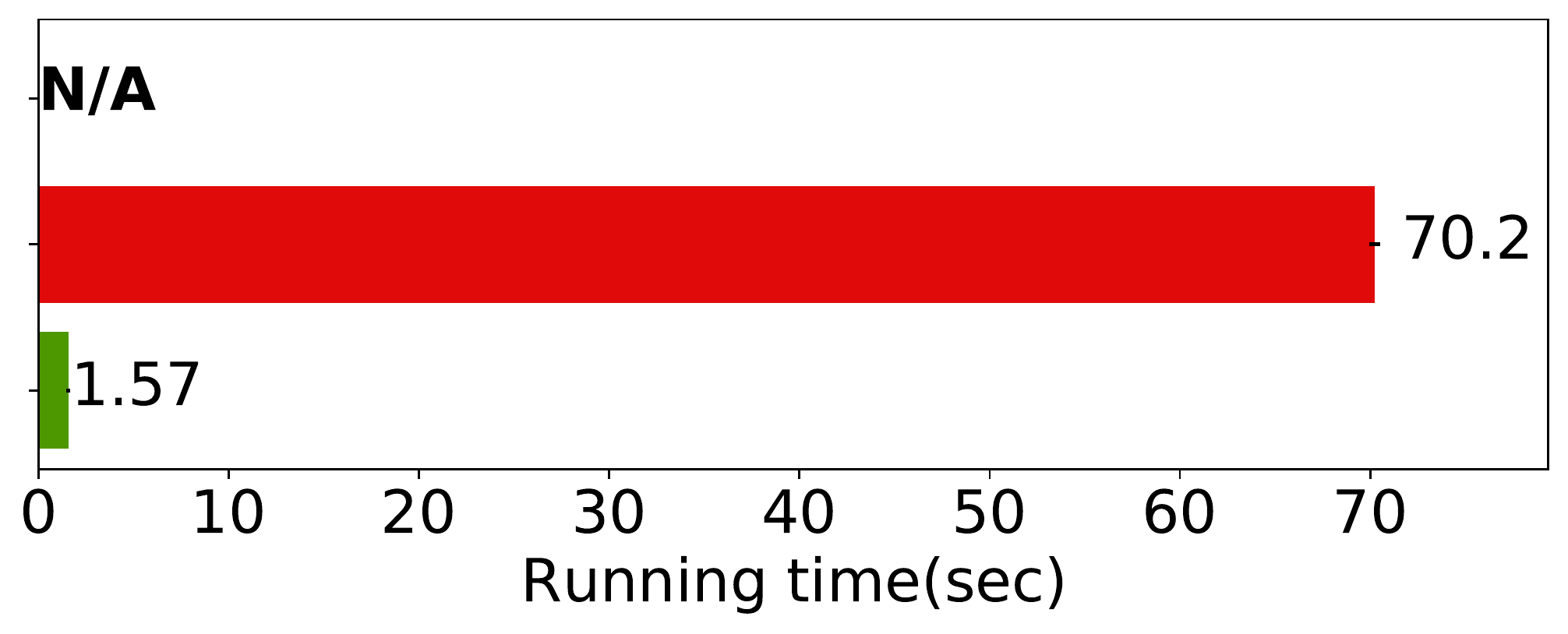}}
\subfloat[Grouping with aggregation-30M]{\includegraphics[width = 0.31\textwidth]{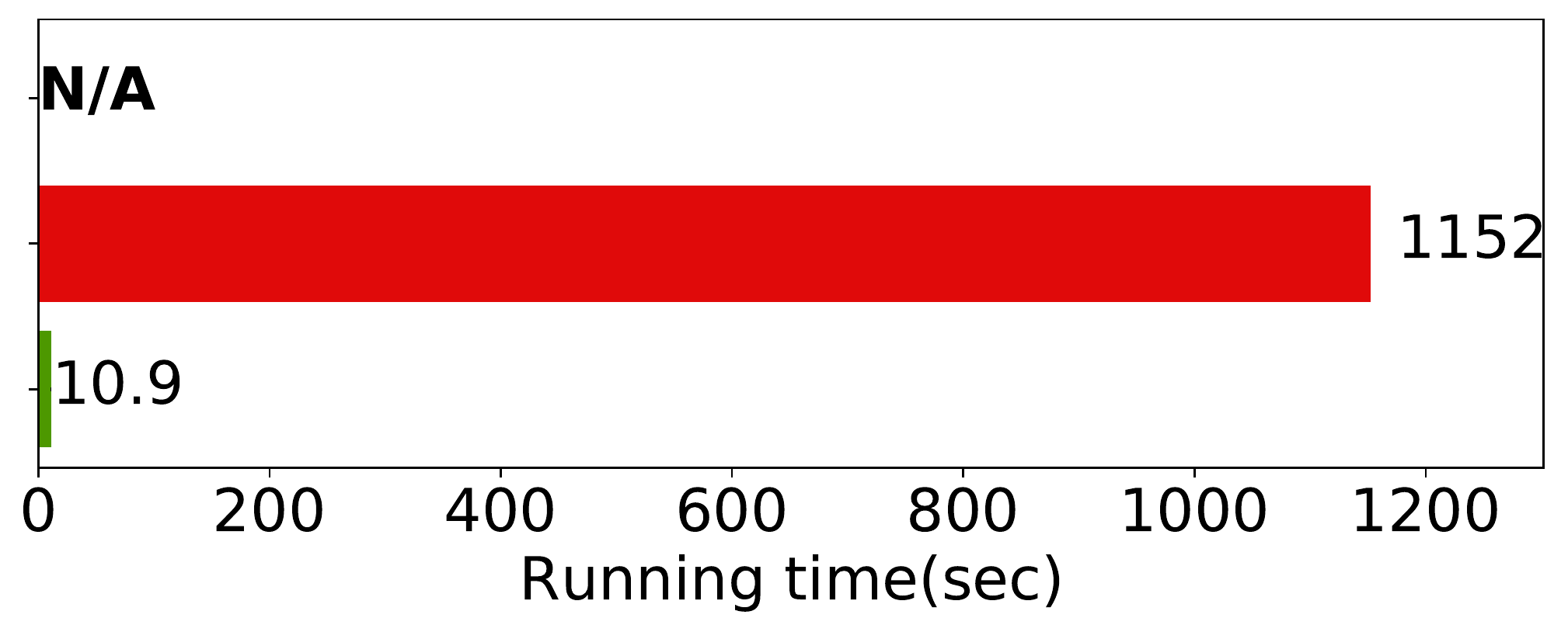}} 
\caption{{Grouping with aggregation operator performance under different input sizes. Error bars show the standard deviations.} }
\label{fig:BDB2}

\end{figure*}

\begin{figure*}[!h]
\centering
\subfloat[Foreign key join-300K]{\includegraphics[width = 0.37\textwidth]{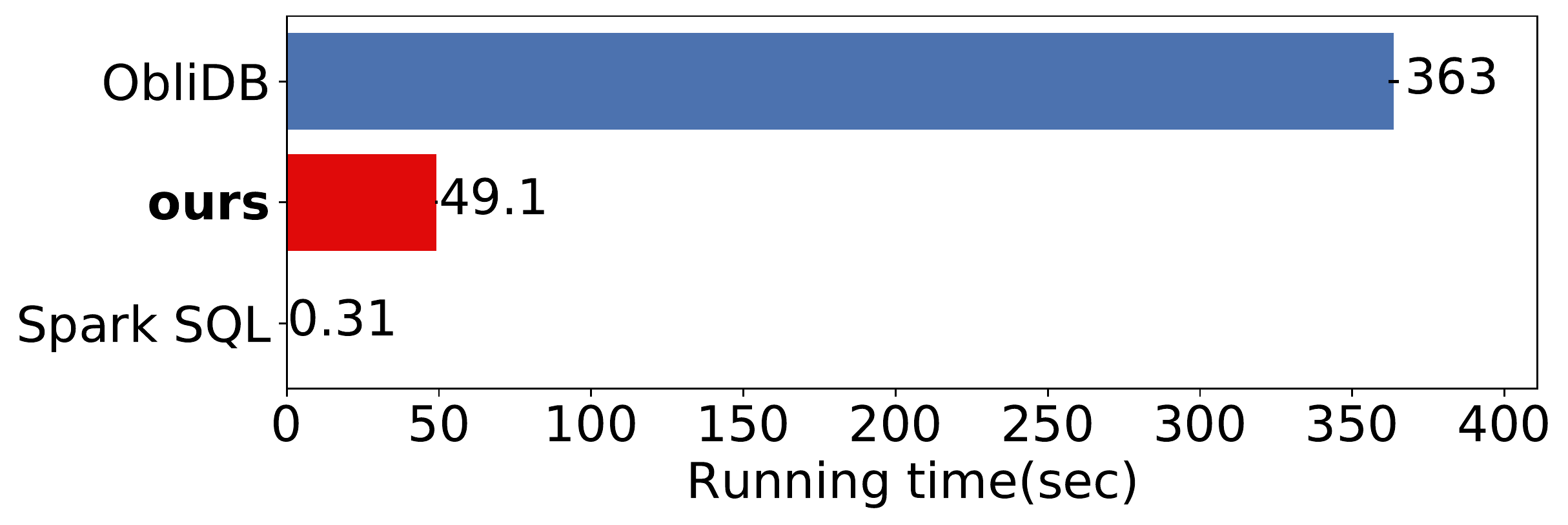}}
\subfloat[Foreign key join-900K ]{\includegraphics[width = 0.31\textwidth]{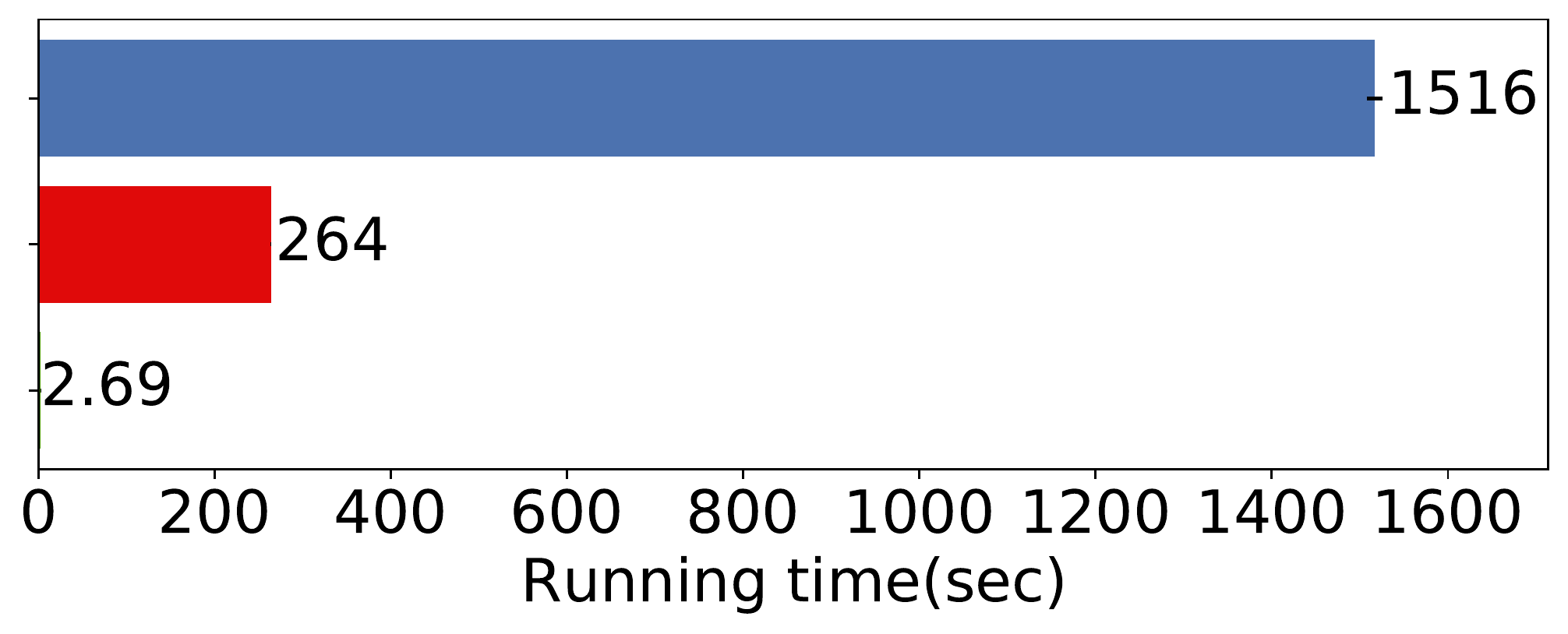}} 
\subfloat[Foreign key join-3M]{\includegraphics[width = 0.31\textwidth]{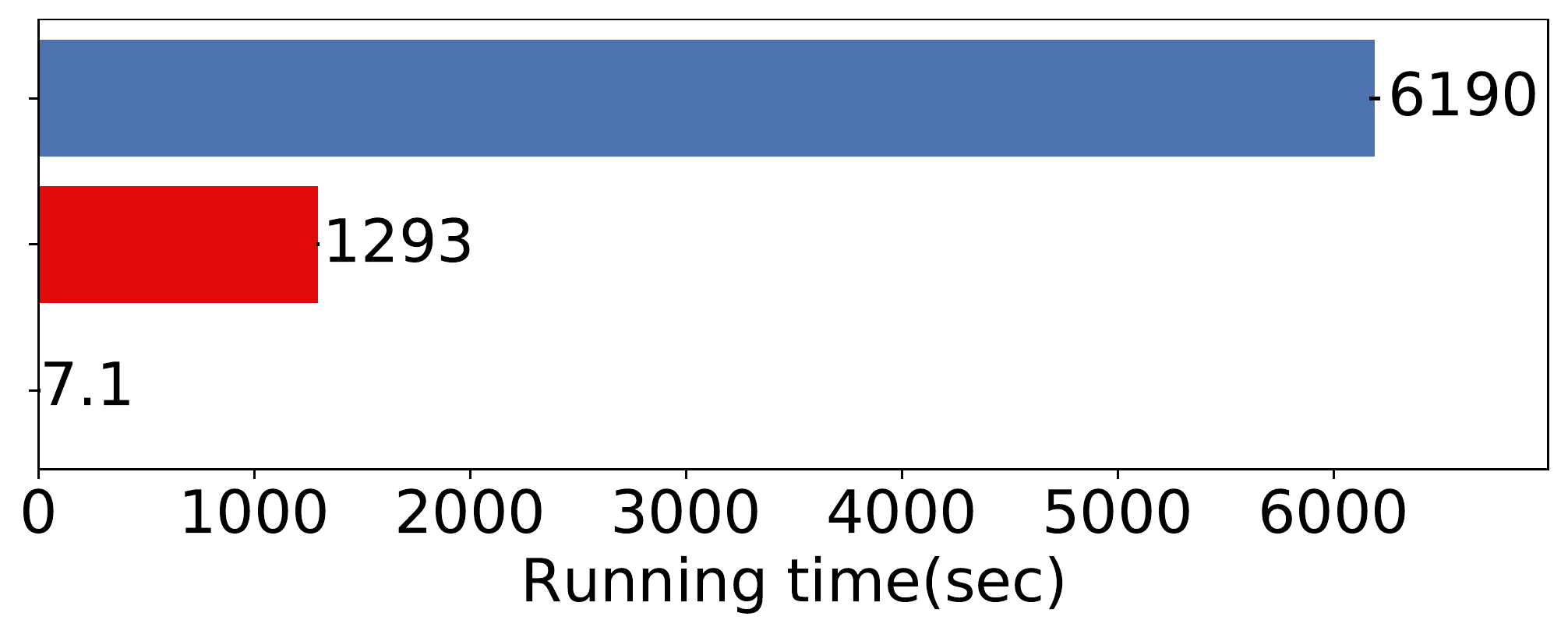}} 
\caption{{ Foreign key join operator performance under different input table sizes. Error bars show the standard deviations. }}
\label{fig:BDB3}
\end{figure*}

\begin{figure}[t]
    \centering
    \includegraphics[width = 0.66\columnwidth]{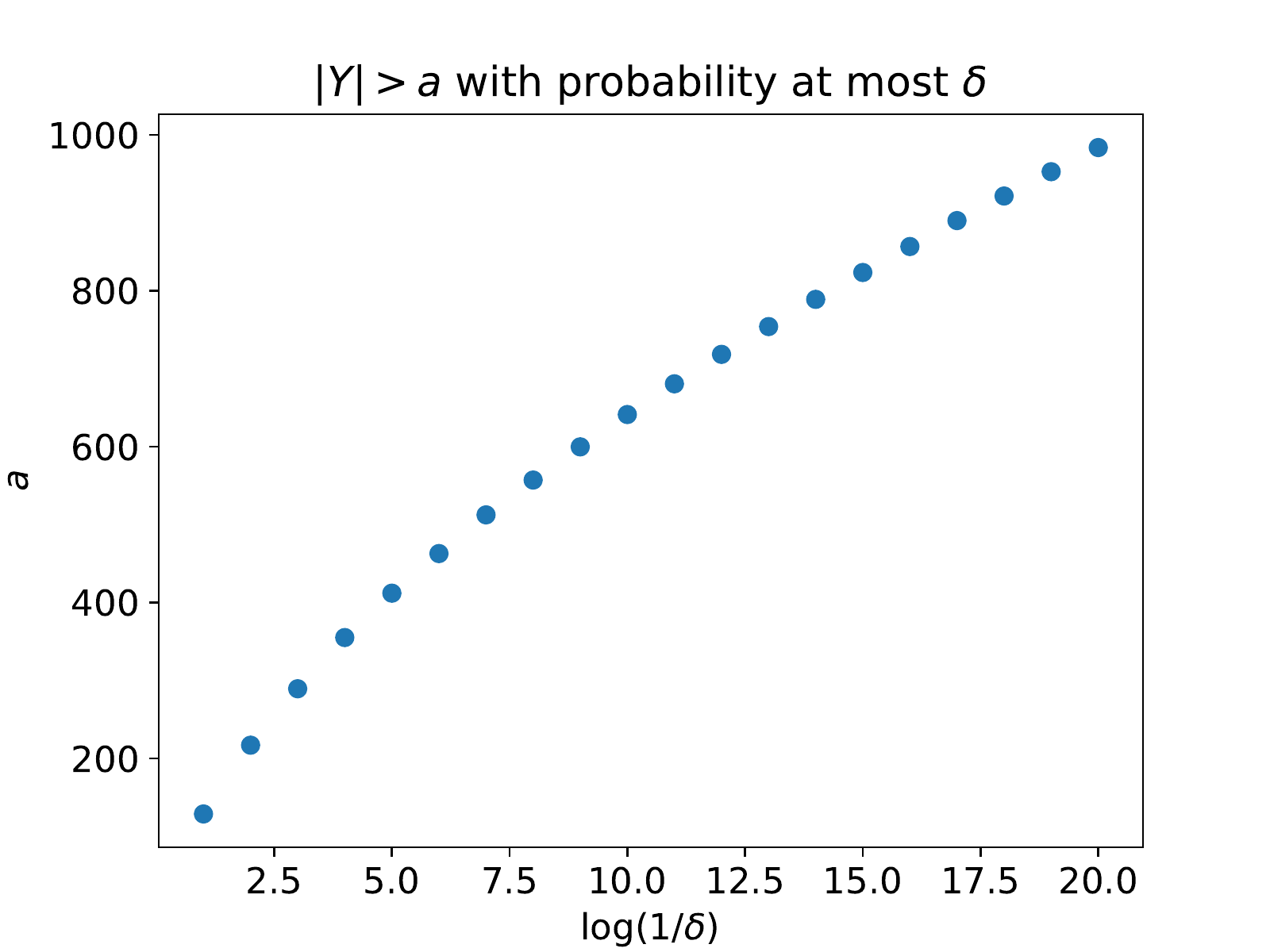}
    \caption{{ Simulated Binary Mechanism Concentration ($\eps = 1$, $N = 10^9 $, each data point uses $10^4/\delta$ trials)}}\label{fig:simulation}
\end{figure}

\noindent\textbf{Setting Privacy Parameters}
We set $\eps$ to $1$ and $\delta$ to $2^{-30}$, which is negligible small (e.g. $\delta < 1/N$, where $N$ is the size of the data). These settings follows the standard 
privacy settings in differentially private systems such as 
PINQ~\cite{pinq}, Vuvuzela~\cite{Vuvuzela}, and RAPPOR~\cite{rappor}. 
To further optimize the privacy parameters, we use numeric simulation to calculate a tighter bound of the differentially private mechanism when possible. For example, for the binary mechanism \cite{css10} that we used as a DP oracle in \autoref{alg:dofilter}, 
we can simulate its approximation error by repeating random trials of sum of Laplace noises.
\autoref{fig:simulation} shows the simulation 
result. We can observe that the sum of independently sampled Laplace noises grows 
linearly as the $\frac{1}{\delta}$ grows exponentially. We can estimate the error by assuming linear growth of $|Y|$ over $\frac{1}{\delta}$'s exponentially growth when $\delta$ is too small to simulate.

\subsection{Comparison to Prior Work}
We now evaluate our three DO operators: selection with projection, grouping with aggregation, and foreign key join. 
The Big Data Benchmark~\cite{bdb} directly contains benchmarks to evaluate selection with projection and grouping with aggregation. 
\vldbrevision{Our DO operators only insert dummy tuples to the results to ensure that our memory access pattern is differentially private and do not add noise to the query results. Hence the accuracy of returned results is not affected.}

\vspace{3mm}
\noindent\textbf{Benchmark \#1: Selection with projection:}\\
{\tt
{\color{thegreen}SELECT} pageURL, pageRank \\
{\color{thegreen}FROM} rankings \\
{\color{thegreen}WHERE} pageRank > {\color{theblue}1000}\\
}

The first benchmark performs a selection with projection on \textit{rankings} table.
We compare the performance of ours with Spark
SQL, ObliDB in \autoref{fig:BDB1}, \autoref{fig:BDB2}, and \autoref{fig:BDB3}.
Compared with non-encrypted and non-oblivious spark SQL, for moderately 
large size datasets, ours exhibits $7.8 - 26.0$x overhead.
As shown in 
\autoref{sec:breakdown}, our overhead mostly comes from encryption and decryption when moving data in and out of SGX enclave memory.
The performance gain comes from the batched read and write implemented in our system.
Compared with oblivious systems, we are $1.5 - 3.7\times$ faster than ObliDB in benchmark 1.
This performance gain comes from more efficient algorithm and the less padding size brought by the differential obliviousness.

\vspace{3mm}
\noindent\textbf{Benchmark \#2: Grouping with aggregation:}\\
{\tt
{\color{thegreen}SELECT} SUBSTR(sourceIP, {\color{theblue}1}, {\color{theblue}8}), {\color{thegreen}SUM}(adRevenue)\\
{\color{thegreen}FROM} uservisits \\
{\color{thegreen}GROUP BY} SUBSTR(sourceIP, {\color{theblue}1}, {\color{theblue}8})\\
}
The second benchmark aggregates the sum of \textit{adRevenue} based on their \textit{sourceIP} column over \textit{UserVisits} table.
Compared with non-encrypted and non-oblivious spark SQL, for moderately 
large size datasets(300K - 30M), ours exhibits $12.4 - 105.7$x overhead. 
For grouping with aggregation over \textit{UserVisits} table of 300K rows, ours has similar performance with ObliDB. However, ObliDB's grouping operator assumes that the aggregation statistics of all the distinct groups (up to 400,000 under current SGX enclave memory capacity) can fit in enclave memory so that it can calculate the aggregation results in just one pass, but this assumption does not hold for \textit{UserVisits} table of 3 million rows and more. ObliDB fails to run grouping with aggregation over \textit{UserVisits} table of 30 million rows under our hardware settings too. As the number of distinct groups grows, ours has to process aggregation query in more passes, which is another source of overhead to achieve differentially oblivious grouping with aggregation.

BDB does not contain a benchmark that directly evaluate foreign key join. BDB has a complex benchmark that requires composing a series of database operators. Composing DO operators is beyond the scope of this paper. Here, we use a simplified benchmark to evaluate foreign key join.

\vspace{3mm}
\noindent\textbf{Benchmark \#3: Foreign key join:}\\
{\tt
{\color{thegreen}SELECT} * \\
{\color{thegreen}FROM} Rankings {\color{thegreen}AS} R, UserVisits {\color{thegreen}AS} UV\\
{\color{thegreen}WHERE} R.pageURL = UV.destURL\\
}

The third benchmark is to do foreign key join between \textit{Rankings} and \textit{UserVisits}.
Ours exhibits $98.1 - 182.1$x overhead over Spark SQL, but it is $4.8  - 7.4 \times$ faster than ObliDB in benchmark 3.
As stated before, this performance gain mainly comes from less dummy writes to achieve differential obliviousness compared to full obliviousness.
Bucket oblivious sort achieves $O((N/B) \log_{M/B} (N/B))$ number of page swaps and bitonic sort requires $O(N \log^2 N)$ page swaps if implemented naively. The practical speedup we see is $5-7\times$ partly because bitonic has a smaller constant in the big-$O$.

\subsection{Latency Breakdown}
\label{sec:breakdown}
\begin{figure}[t]
    \centering
    \includegraphics[width=0.48\textwidth,height=\textheight,keepaspectratio]{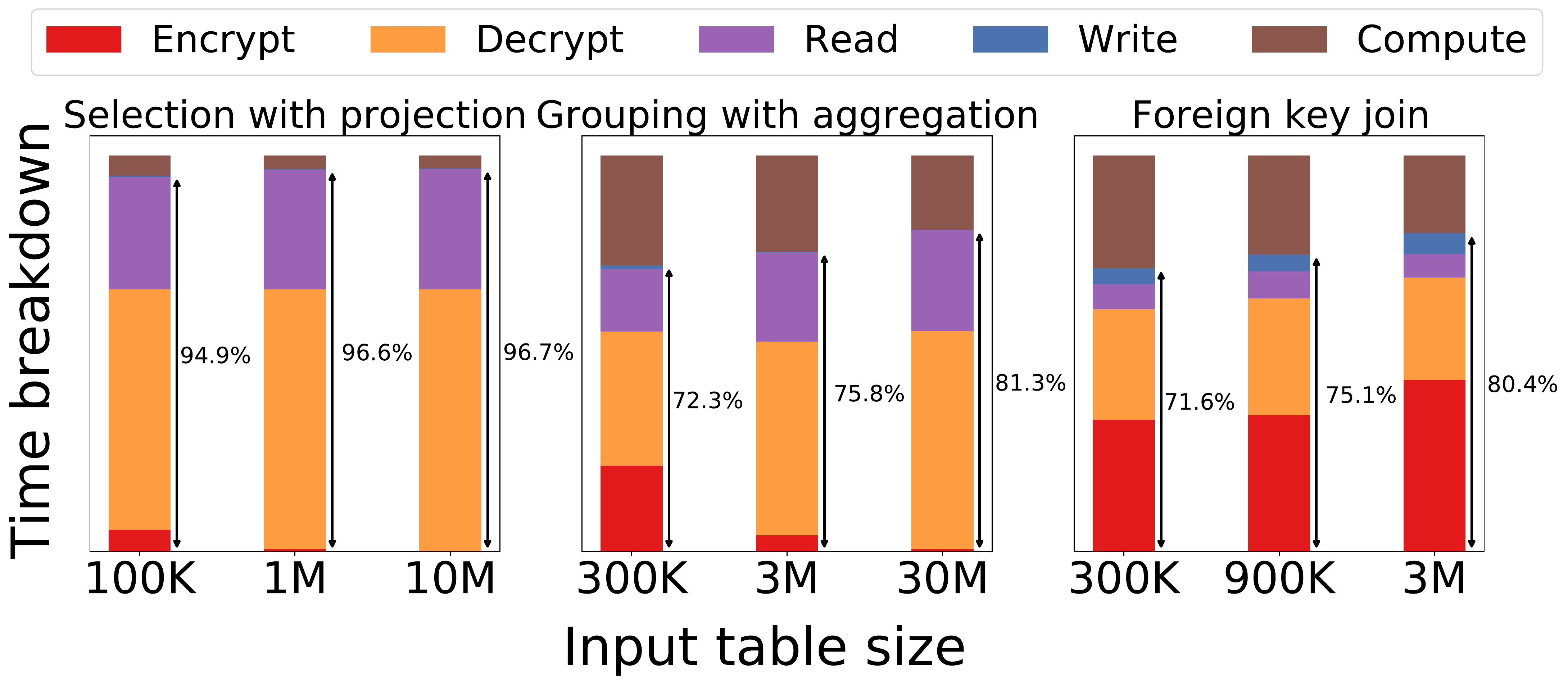}\\
    \caption{{\small DO operators performance breakdown}}
\label{fig:operator-breakdown}
\end{figure}

It is interesting to understand where are the key performance bottlenecks in our differentially oblivious operators.
We break down each of our basic operator's completion time into six categories: (1) decryption within enclave; (2) encryption within enclave; (3) reading from untrusted memory to enclave buffer; (4) writing from enclave buffer to untrusted memory; and (5) computation within enclave. \vldbrevision{We use Read Time-Stamp Counter (RDTSC) to profile the time spent on each category.}
\autoref{fig:operator-breakdown} shows the performance breakdown results. Encryption, decryption, memory copy between untrusted memory and enclave memory are the major overheads in our differentially oblivious operators. Under the largest input table size scenario, the real query computation time only accounts for 2\% of the total execution time of the filter operator. Memory copy between untrusted memory and enclave memory accounts for 30\% of the total time and encryption plus decryption take up the rest 68\%.
Because applying hash function to distinguish different groups in grouping with aggregation operator and oblivious sorting in foreign key join operator are more expensive than the simple comparison in filter operator, the computation constitutes a larger fraction in the operator execution time. 

\vldbrevision{For the encryption time portion, we find that it decreases as the input table size increases in \textsc{DoGroup}$_h$. This is because the number of groups increases sublinearly as the size of input table increases when we group by the first $8$ bytes of IP address in BDB2 benchmark. The encryption time portion increases as the input table size increases in \textsc{DoJoin}, because bucket oblivious sort is the dominant overhead and its cache complexity is $O((N/B) \log_{M/B} (N/B))$ and the overhead from compute part grows linearly to the size of primary and foreign key tables.}

\vldbrevision{Moving data between the enclave and untrusted memory incurs overhead from both SGX ECALL/OCALL and encryption and decryption within the enclave. }
These results validate that the \vldbrevision{data movement} between the trusted and untrusted components is the key bottleneck for our operators, which justify the usage of cache complexity and output size as the key theoretical performance metrics for the analysis of oblivious operators.

\begin{table}[!ht]
\caption{Time percentage spent on inserting padding tuples.}
\vldbrevision{
\begin{tabular}{|c|c|c|c|}
\hline
Dataset Size & \textsc{DoFilter} & \textsc{DoGroup}$_h$ & \textsc{DoJoin} \\ \hline
Small        & 51.0\%    & 28.5\%     & 26.3\%  \\ \hline
Medium       & 12.5\%    & 2.9\%     & 40.1\%  \\ \hline
Large        & 1.7\%    & 0.4\%     & 47.5\% \\ \hline
\end{tabular}

\label{tab:padding_time}
}
\end{table}

\vldbrevision{
\textbf{Cost of tuple padding procedure.} We further profile how much percentage of overall execution time is spent for inserting tuple padding. As shown in Table~\ref{tab:padding_time}, for \textsc{DoFilter} and \textsc{DoGroup}$_h$ operators, the time portion spent on tuple padding decreases when the size of input table grows, and the time portion is smaller than $2\%$ when the input table size is large. The reason is that \textsc{DoFilter} and \textsc{DoGroup}$_h$ do not need to do worst-case padding like full oblivious guarantee and the ratio of number of padding tuples and size of input decreases when the size of input table becomes larger. For \textsc{DoJoin}, the percentage of tuple padding procedure grows when when the size of input table increases, because in the bucket oblivious sort at the end of each iteration of oblivious random bucket assignment, we need to pad each bucket with dummy tuples until full. The number of padding tuples needed in oblivious random bucket assignment grows faster than linear.
}

%% file: relat.tex
\section{Related Work}\label{sec:related_work}

\paragraph{Encrypted Databases.}
There are a series of encrypted database systems uses standard or customized encryption schemes. For example,
CryptDB~\cite{popa2011cryptdb} uses a multi-layer encryption scheme 
to allow user to set different security levels for different columns. 
Arx~\cite{arx} uses strong encryption and applies special data structures to enable search. 
Other systems~\cite{Bost16ccs, BostMO17CCS, PPDG16SIGMOD, peng2020falcondb } build on searchable encryption techniques.
All these systems only encrypt data, not access patterns. As a result, they are all vulnerable to access pattern attacks.
Recently, there are 
many new database systems based on hardware enclaves, such as 
TrustedDB~\cite{bajaj2013trusteddb}, Cipherbase~\cite{cipherbase},
EnclaveDB~\cite{priebe2018enclavedb}, VC3~\cite{schuster2015vc3}, VeriDB~\cite{zhou2021veridb} and StealthDB\cite{gribov2017stealthdb}. These systems all leave data outside enclaves encrypted. However, these systems either only support data that can fit into very limited enclave memory (128MB in case of Intel SGX), such as EnclaveDB, or vulnerable to memory access pattern attacks.

\paragraph{Oblivious Databases.}
To address the vulnerability to access pattern attacks, 
recent data analytic systems like Opaque~\cite{zheng2017opaque} and 
ObliDB~\cite{eskandarian2017oblidb} proposed and implemented a few database query processing algorithms that are fully oblivious. However, there are significant performance penalties of their oblivious modes compared 
to the non-oblivious or partial-oblivious (but encrypted) counter-parts.
Obladi~\cite{crooks2018obladi} focuses on providing ACID transactions; federated oblivious database systems ~\cite{bater2017smcql, conclave, shrinkwrap,oblivious-coopetitive-analytics} provide cooperative data analytics for untrusted parties (semi-honest or malicious). Oblix~\cite{mishra2018oblix} is an oblivious search index whose internal memory access is also oblivious.
Shrinkwrap \cite{shrinkwrap} uses fully oblivious operators but padding with DP guarantees, and this greatly reduces its intermediate query results sizes. We are the first work to demonstrate the theoretical and empirical performance of differentially oblivious database operators.

\paragraph{ORAM and Oblivious Algorithms.}
Oblivious RAM and oblivious computation were proposed in the seminar work 
by Goldreich \cite{g87}. Since then, various ORAM schemes and hardware implementations 
were proposed, such as Path ORAM~\cite{pathoram}, Ring ORAM \cite{ringoram}, and PrORAM \cite{proram}.
Despite these exciting advances, ORAM still suffers from a $\log(N)$ factor slow down. For database that potentially has billions of tuples, this overhead is significant. 
In addition, using ORAM while leaking the runtime or result length does not provide full obliviousness.
GhostRider~\cite{liu2015ghostrider} provides an FPGA-based implementation to ensure memory-trace obliviousness by employing ORAM.
ZeroTrace~\cite{zerotrace} is a library of oblivious memory primitives for SGX enclave against side-channel attacks.
Obfuscuro~\cite{ahmad2019obfuscuro} leverage ORAM operations to perform secure code execution and data access, and ensures that  the program always runs for a pre-configured time interval
Apart from ORAM, many other oblivious data structures have been proposed, such as oblivious priority queues \cite{Shi20,JafargholiLS19}. 
Apart from differential obliviousness~\cite{ccms19},
Allen et al. \cite{adknoy19} proposed a security model, ODP, which combines differential obliviousness and differential privacy.
This model is useful when both the published result and the memory access pattern need to be protected.

\vldbrevision{
\paragraph{Other Ways of Mitigating SGX Side-channel Vulnerability.} DR.SGX \cite{drsgx19} designs and implements a compiler-based tool that instruments the enclave code, permuting data locations at fine granularity. By periodically re-randomizing all enclave data, DR.SGX can prevent correlation of repeated memory accesses.
T-SGX~\cite{tsgx17} ensures that no page fault sequence will be leaked to attackers via Intel  Transactional Synchronization Extensions (TSX) in order to mitigate the memory side channel attacks.
}

\paragraph{Differential Privacy.}
Another related development is differential privacy.
Since its introduction~\cite{dmns06},
differential privacy has become the de facto standard 
for protecting user privacy. 
Many differential privacy data analytics systems have been developed, such as 
PINQ~\cite{pinq},
FLEX~\cite{flex-dp},
GUPT~\cite{gupt},
PrivateSQL~\cite{kotsogiannis2019privatesql}.
In this paper, we \vldbrevision{use} a differentially private prefix-sum algorithm ~\cite{css10} as a building block of our differentially oblivious filtering algorithm. 
Additionally, we uses two established theoretical results in differential privacy, the group privacy theorem \cite{Vadhan17dpbook} and the basic composition~\cite{dr14}.

%% file: discussion.tex
\section{Discussion}
\label{sec:discussion}
Our paper presents the first step towards using different obliviousness in databases. Although several theoretical papers have already been moving in this direction~\cite{ccms19, chu2021dojoin}, our paper is the first one that have designed and implemented database operators and show their empirical speedup against fully oblivious operators. The result is promising: we show that differentially oblivious operators can deliver up to 7.4 $\times$ performance improvement.
Now, one interesting question is how far away we are from an end-to-end differentially oblivious database. This is admittedly our original goal for the project, however, we have encountered substantial challenges. We want to leave them as future works for the research community.

\paragraph{Operator Composition.} A complex SQL query needs to combine multiple operators. DO operators are defined on two neighboring databases. Let's imagine we want to apply a differentially oblivious operator ${\mathcal M_2}$ to the outcome of another
differentially oblivious operator ${\mathcal M_1}$. To ensure differential obliviousness end-to-end, we need to make sure ${\mathcal M_1}$ is distance-preserving. We say that an operator is distance preserving, iff when applying the operator to two neighboring databases, the two output databases are still neighboring databases.
This is required because ${\mathcal M_2}$'s obliviousness guarantee depends on the inputs to ${\mathcal M_2}$ to be neighboring databases.
It is unclear how to build database operators that are both distance preserving and DO. Operators such as join are particularly challenging because join can increase distance.
To date, there is only a theoretical work \cite{zhou2022theory} that is able to compose DO database operators but there is a long way towards practical DO composability.

\paragraph{Query Optimization.} When we have more differentially oblivious operators in the future (e.g,. sort-based grouping with aggregation, hash-based join) and need to run multiple operators to serve one SQL query, we need to choose which operator to use to accelerate query execution. For example, how do we choose sort-based grouping with aggregation or hash-based grouping with aggregation for a given query, and how to generate the optimal query execution plan will be another interesting problem to solve.

\paragraph{Access Patterns for Private Memory.}
Our design patches the side channel of the access pattern leakage for the public memory. The hardware we implement our algorithms on, Intel SGX, has known vulnerabilities for the access pattern leakage for the private memory, and this can also leak sensitive information.
Specifically, popular commodity processors (even the ones with secure enclaves such as Intel SGX) allow time-sharing of the same 
on-chip cache among different processes.
This leads to a series practical cache-timing attacks
~\cite{bernstein2005cache-aes, RistenpartTSS09, DemmeMWS12, ZhangJRR12, ZhangJRR14}.
Fortunately, we can harden our implementation against cache-timing attacks without dramatic changes. 
The recipe is to make the algorithms and data structures within private memory oblivious as well. 
For example, we can change our implementation of bucket oblivious sort (in \textsc{DoGroup}$_s$ and \textsc{DoJoin}) so that it is oblivious within private memory. We can also use oblivious priority queues such as~\cite{Shi20} to implement the priority queue in \autoref{alg:dpcount}. 
Our comparison with ObliDB is fair: both ObliDB and our implementation do not consider private enclave memory access pattern leakage. 

%% file: conclusion.tex
\section{Conclusion}\label{sec:conclusion}

Preventing data leakage in cloud databases has become a critical problem.
Leveraging secure execution in hardware enclaves, such as Intel SGX, is not enough to prevent an attacker from breaking data confidentiality by observing the access patterns of encrypted data. 
Ensuring oblivious access patterns can lead to substantial performance overheads.
In this paper, we study how to incorporate into databases one new notion of obliviousness, $\mathit{differential~obliviousness}$, a novel obliviousness property which ensures that memory access patterns satisfy differential privacy.
We design and implement \textbf{Adore}: \textbf{A} set of \textbf{D}ifferentially \textbf{O}blivious \textbf{RE}lational database operators, and we formally prove that they satisfy the notion of differential obliviousness.
Our evaluations show that our differentially oblivious operators outperform the state-of-the-art fully oblivious databases by up to $7.4\times$ on Big Data Benchmark dataset with the same hardware configuration. 

\section*{Acknowlegments}

The authors would like to thank Bolin Ding, Cong Yan, Derek Leung for helpful discussions and the valuable suggestions from anonymous reviewers.
This work is supported by an NSF award 2128519 and an ONR grant N000142212064.

%% file: app.tex
\vldbrevision{
\paragraph{Roadmap.} In Section~\ref{sec:tighter_bound} we provide a tighter bound for binary mechanism of deferentially private prefix sum. In Section~\ref{sec:dp_distinct:app} we present our differentially private distinct count algorithm design and analysis. In Section~\ref{sec:rand_partition} we prove some properties of binomial distribution. 

}

\section{A Tighter Bound for Binary Mechanism of DP Prefix-Sum}\label{sec:tighter_bound}

For a Laplace random variable, we have the following fact
\begin{fact}[Basic facts for Laplace random variable]\label{fact:lap1}
Let $x$ denote a random variable sampled from ${\sf Lap}(b)$, then 
\begin{align*}
    \E[x^2] = 2b^2, 
\end{align*}
and with probability $1 - \delta$, we have
\begin{align*}
    | x | \leq  b \ln ( 1 / \delta )
\end{align*}
\end{fact}
\begin{proof}
The CDF of ${\sf Lap}(b)$ is:
\begin{align*}
   F(x)=  \begin{cases}
    \frac{1}{2} \exp (  x / b ), &  \text{if } x < 0 ; \\
    1 - \frac{1}{2} \exp ( - x / b ), & \text{if } x \geq 0. \\
    \end{cases}
\end{align*}
we have:
\begin{align*}
    F( b \ln(1/\delta) ) = 1 -  \delta / 2.
\end{align*}
This fact is proved by the symmetry of Laplace distribution. 
\end{proof}

\begin{lemma}\label{lem:lap2}
Let $x_1, x_2, \ldots, x_n$ be $n$ i.i.d. random variables sampled from Laplace distribution 
${\sf Lap}(b)$. For $0 \leq \delta \leq 1$, with at most probability $\delta$, we have for all $i \in [n]$:
\begin{align*}
    | x_i | > b \ln( n / \delta ) .
\end{align*}
\end{lemma}
\begin{proof}
For Fact~\ref{fact:lap1}, we know for each $i \in [n]$, with probability $p_i = \delta / n $:
\begin{align*}
    | x_i | > b \ln ( n / \delta )
\end{align*}
Then, for all $i$, using Boole's inequality:
\begin{align*}
    \Pr [ | x_i | > b \ln ( n / \delta ) ] \leq  \sum_{i=1}^n p_i = \delta .
\end{align*}
\end{proof}

\begin{lemma}[Bernstein inequality \cite{b24}]\label{lem:bernstein}
Let $X_1, \cdots, X_n$ be independent zero-mean random variables. Suppose that $|X_i| \leq M$ almost surely (with probability $1-\delta$), for all $i$. Then, for all positive $t$,
\begin{align*}
\Pr \left[ \sum_{i=1}^n X_i > t \right] \leq \exp \left( - \frac{ t^2/2 }{ \sum_{j=1}^n \E[X_j^2]  + M t /3 } \right) + \delta.
\end{align*}
\end{lemma}

\begin{lemma}[similar to Lemma 2.8 in \cite{css11}, sum of independent Laplace distributions]\label{lem:lap3}
Let $x_1, x_2, \cdots, x_n$ denote $n$ i.i.d. random variables sampled from Laplace distribution ${\sf Lap}(b)$. For all $t > 0$, we have
\begin{align*}
    \Pr \left[ \sum_{i=1}^n x_i > t \right] \leq \exp \Big( - \frac{t^2 / 2}{  2 n b^2 + M t /3} \Big) + \delta
\end{align*}
where $M =  b \ln( n / \delta )$.
\end{lemma}
\begin{proof}

The proof is directly from Bernstein inequality and definition of Laplace distribution.
From Lemma~\ref{lem:lap2}, we have for all $i \in [n]$
\begin{align*}
    |x_i| \leq b \ln ( n / \delta)
\end{align*}
holds with at least probability $1-\delta$. 
\end{proof}

\begin{lemma}[a slightly tighter version of Collary 2.9 in \cite{css11}, measure concentration]
Let $x_1, x_2, \ldots, x_n$ be $n$ i.i.d. random variables sampled from Laplace distribution ${\sf Lap}(b)$. For any $0 \leq \delta \leq 1$, we have:
\begin{align*}
    \Pr [ \sum_{i=1}^n x_i > \max \Big\{ \sqrt{4 n b^2 \ln( 3/ \delta) }, \\
    (2/3) b \ln( {3 n}/ {\delta} ) \cdot \ln ( {3}/{\delta}) \Big\} ] 
    \leq  \delta
\end{align*}
\end{lemma}
\begin{proof}
From Lemma~\ref{lem:lap3}, we have for $0 \leq \delta / 3 \leq 1$ ($M = b \ln( 3 n / \delta) $)
\begin{align*}
    \Pr \left[ \sum_{i=1}^n x_i > t \right] \leq  &
    \exp \left( - \frac{t^2 / 2}{  2 n b^2 + b \ln ( 3 n / \delta )  \cdot t /3} \right) + \delta / 3 \\
    = & \exp \left( - \frac{t^2 }{  4 n b^2 + (2/3) b \ln ( 3 n/ \delta) \cdot  t} \right) +  \delta / 3
\end{align*}
Now we prove an inequality for the general form of the exponent, 
for $a > 0, b > 0, t > 0, k > 0, t = \max \{ \sqrt{ak}, bk \}$:
\begin{align*}
\frac{t^2}{c + b t}   \geq \frac{t^2}{2 \max \{ a, bt \} }  = \frac{\max \{ ak, bkt \} }{2 \max \{ a, bt \} } = \frac{k}{2} .
\end{align*}
We choose $t$ as follows,
\begin{align*}
    t = \max \{ \sqrt{4 n b^2 \ln( 3 / \delta) }, (2/3) b \ln ( 3 n / \delta ) \ln( 3 / \delta ) \},
\end{align*}
then we have:
\begin{align*}
    \Pr \left[ \sum_{i=1}^n x_i > t \right] \leq \exp ( - 0.5 \ln( 3/ \delta ) ) + \delta/3 = \delta .
\end{align*}
\end{proof}

%% file: appendix_dp_distinct.tex
\section{Differentially Private Distinct Count} \label{sec:dp_distinct:app}
In this section, we describe a differentially private distinct count algorithm based on \cite{bjkst02}. We first prove the main technical lemmas 
about order statistics properties of random sampling in \autoref{sec:tech_lemma}. Next, we define $(\eps, \delta)$-sensitivity and 
introduce the Laplacian mechanism for $(\eps, \delta)$-sensitivity in \autoref{sec:newsensitivity}. This follows by our analysis of \cite{bjkst02}: its $(\eps, \delta)$-sensitivity and its 
approximation ratio concentration.
Last, we develop a differentially private distinct count algorithm based
on \cite{bjkst02} (Algorithm~\ref{alg:dpcount:app}) and also demonstrate a simpler 
$1.1$ approximation version (Algorithm~\ref{alg:dpcount2:app}).

\subsection{Order Statistics Properties of Random Sampling}\label{sec:tech_lemma:app}

\begin{claim}[Restatement of Claim~\ref{claim:1}]\label{claim:1:app}
Let $t \in [n]$, and fix any $0 < \delta < 1/2$. Then we have the following two bounds:
\begin{enumerate}
    \item $\Pr[ y_t > \delta \frac{t}{n} ] \geq 1 - \delta$. 
    \item $\Pr[ y_t >  \frac{t}{2n} ] \geq 1 - \exp(-t/6)$. 
\end{enumerate}

\end{claim}
\begin{proof}

{\bf Part 1.} 
Consider the interval $[0,\delta \frac{t}{n}]$. We have 
\begin{align*}
    \ex{ | \{ i \in [n] : x_i < \delta t/n \} | } = \delta t
\end{align*} namely, the expected number of points $x_i$ will fall in this interval is exactly $\delta t$. Then by Markov's inequality, we have 
\begin{align*}
    \Pr[ | \{ i \in [n] : x_i < \delta t/n \} | \geq t] \leq \delta .
\end{align*} 
So with probability $1-\delta$, we have $| \{ i : x_i < \delta t/n \} | < t$, and conditioned on this we must have $y_t > \delta t/n$ by definition, which yields the first inequality.

{\bf Part 2.} 
For the second inequality, note that we can write  
\begin{align*}
    | \{ i \in [n] : x_i <  t/(2n) \} | = \sum_{i=1}^n z_i
\end{align*}
where $z_i \in \{0,1\}$ is an random variable that indicates the event that $x_i <  t /2n$. Moreover, $\ex{\sum_{i=1}^n z_i} = t/2$. Applying Chernoff bounds, we have
\begin{align*}
    \Pr[ | \{ i \in [n] : x_i <  t/(2n) \} | \geq t ] \leq \exp( - t / 6 )    
\end{align*}
Which proves the second inequality.
\end{proof}

\begin{claim}[Restatement of Claim~\ref{claim:2}]\label{claim:2:app}
Fix any $4 <\alpha <n/2$, and $1 \leq t \leq n/2$. Then we have
\begin{align*}
  \Pr[  y_{t+1} < y_t + \alpha/n ] \geq 1- \exp( - \alpha / 4 ) .
\end{align*}
\end{claim}
\begin{proof}\label{proof:b.2}
Note that we can first condition on any realization of the values $y_1,y_2,\dots,y_t$ one by one. Now that these values are fixed, the remaining distribution of the $(n-t)$ uniform variables is the same as drawing $(n-t)$ uniform random variables independently from the interval $[y_t,1]$. Now observed that for any of the remaining $n-t$ uniform variables $x_i$, the probability  that $x_i \in [y_t, y_t + \alpha/n]$ is at least $\frac{\alpha}{n}$, which follows from the fact that $x_i$ is drawn uniformly from $[y_t,1]$. Thus, 
\begin{align*}
        & ~ \Pr[ | \{  i  \in S : x_i \in [y_t, y_t + \alpha/n] \}| = 0 ] \\
       \leq & ~ (1 - {\alpha}/{n} )^{n-t} \\ 
       \leq & ~ (1 - {\alpha}/{n} )^{n/2} \\ 
        = & ~ \exp( \frac{n}{2} \log( 1 - \alpha / n ) ) \\ 
        < & ~ \exp( \frac{n}{2} (- \alpha/n + 2 (\alpha/n)^2 ) ) \\ 
        < & ~ \exp( - \frac{n}{2} \frac{ \alpha }{ 2n } )\\ 
        = & ~ \exp( - \alpha / 4 ). 
\end{align*}

Thus $| \{  i  \in S : x_i \in [y_t, y_t + \alpha/n] \}| \geq 1$ with probability at least $1-e^{-\alpha/4}$. Conditioned on this, we must have $y_{t+1} < y_t + \alpha/n$, as desired.

\end{proof}

\begin{lemma}[Restatement of Lemma~\ref{lem:main}]\label{lem:main:app}
Fix any $0 < \beta \leq 1/2$, $1 \leq t \leq n/2$, and $\alpha$ such that $4 < \alpha < \beta t/2$. Then we have the following two bounds:

\begin{enumerate}
    \item $\Pr[ | \frac{1}{y_t} - \frac{1}{y_{t+1}}| < \frac{\alpha}{\beta^2} \frac{n}{t^2} ] \geq 1-\beta - \exp( - \alpha / 4 ) $.
    \item $ \Pr[ | \frac{1}{y_t} - \frac{1}{y_{t+1}}| < 4 \alpha \frac{n}{t^2} ] \geq 1- \exp( - t / 6 ) - \exp( - \alpha / 4 ) $.
\end{enumerate}
\end{lemma}
\begin{proof}

{\bf Part 1.}
For the first statement, we condition on $y_t > \beta \frac{t}{n}$ and, $y_{t+1} < y_t + \frac{\alpha}{n}$, which by a union bound hold together with probability $1-\beta - e^{-\frac{\alpha}{4}}$ by Claims \ref{claim:1:app} and \ref{claim:2:app}. Define the value $t'$ such that $y_t = \frac{t'}{n}$. By the above conditioning, we know that $t' > \beta t$. Conditioned on this, we have 
\begin{align}\label{lem:main_lem_eq1:app}
        \Big| \frac{1}{y_t} - \frac{1}{y_{t+1}} \Big| 
        < & ~ \frac{n}{t'} - \frac{1}{t'/n + \alpha/n} \notag \\
        = & ~ \frac{n}{t'} - \frac{n}{t' + \alpha} \notag \\
        = & ~ \frac{n}{t'}\left(1- \frac{1}{1 + \alpha/t'} \right) \notag \\
        < & ~ \frac{n}{t'}\left(1- (1 - \alpha/t') \right) \notag \\
        \leq & ~ \frac{\alpha n}{(t')^2} \notag\\
        \leq & ~ \frac{\alpha n}{\beta^2 t^2}
\end{align}
Where we used that $\alpha / t' < \alpha /( \beta t) < 1/2$, and the fact that $1/(1+x) >1-x$ for any $x \in (0,1)$.

{\bf Part 2.}
For the second part, we condition on $y_t >  \frac{t}{2n}$ and, $y_{t+1} < y_t + \frac{\alpha}{n}$, which by a union bound hold together with probability $1- e^{-\frac{t}{6}}- e^{-\frac{\alpha}{4}}$ by Claims \ref{claim:1:app} and \ref{claim:2:app}. Then from Lemma~\ref{lem:main_lem_eq1:app}:
\begin{align*}
    \Big| \frac{1}{y_t} - \frac{1}{y_{t+1}} \Big| & \leq  ~ \frac{\alpha n}{\beta^2 t^2} \\
    & = ~4 \alpha \frac{n}{t^2}
\end{align*}

In this case, the same inequality goes throguh above with the setting $\beta = 1/2$, which finishes the proof.
\end{proof}

\subsection{\texorpdfstring{$(\eps,\delta)$}{}-Sensitivity}\label{sec:newsensitivity:app}

In what follows, let $\mathcal{X}$ be the set of databases, and say that two databases $X,X' \in \mathcal{X}$ are neighbors if $\|X-X'\|_1 \leq 1$.
\begin{definition}
Let $f: \mathcal{X} \to \R$ be a function. We say that $f$ has sensitivity $\ell$ if for every two neighboring databases $X,X' \in \mathcal{X}$, we have $|f(X) - f(X')| \leq \ell$.
\end{definition}
\begin{theorem}[The Laplace Mechanism \cite{dmns06}]
Let $f: \mathcal{X} \to \R$ be a function that is $\ell$-sensitive. Then the algorithm $A$ that on input $X$ outputs $A(X) = f(X) +\text{Lap}(0,\ell/\eps)$ preserves $(\eps,0)$-differential privacy.
\end{theorem}

In other words, we have $\Pr[ A(X) \in S ] = (1 \pm \eps) \Pr[ A(X') \in S ]$ for any subset $S$ of outputs and neighboring data-sets $X,X' \in \mathcal{X}$. Now consider the following definition. \vldbrevision{We now introduce a small generalization of pure sensitivity (Definition \ref{def:puresensitive}), that allows the algorithm to \textit{not} be sensitive with a very small probability $\delta$. The difference between $\ell$-sensitivity and $(\ell,\delta)$-sensitivity is precisely analogous to the difference between $\eps$-differential privacy and $(\eps,\delta)$-differential privacy, where in the latter we only require the guarantee to hold on a $1-\delta$ fraction of the probability space. Thus, to achieve $(\eps,\delta)$-differential privacy (as is our goal), one only needs the weaker $(\ell,\delta)$ sensitivity bounds.}

\begin{definition}[$(\ell,\delta)$-sensitive]\label{def:generalsensitive:app}
Fix a randomized algorithm $\mathcal{A}: \mathcal{X} \times R \to \R$ which takes a database $X \in \mathcal{X}$ and a random string $r \in R$, where $R = \{0,1\}^m$ and $m$ is the number of random bits used. We say that $\mathcal{A}$ is $(\ell,\delta)$-sensitive if for every $X \in \mathcal{X}$ there is a subset $R_X \subset R$ with $|R_X| > (1-\delta)|R|$ such that for any neighboring datasets $X,X' \in \mathcal{X}$ and any $r \in R_X$ we have $|\mathcal{A}(X,r) - \mathcal{A}(X',r)| \leq \ell$
\end{definition}
Notice that our algorithm for count-distinct is $(O(\alpha\frac{n}{t}),O(e^{-t} + e^{-\alpha}))$-sensitive, following from the technical lemmas proved above. We now show that this property is enough to satisfy $(\eps,\delta$)-differential privacy after using the Laplacian mechanism.

\begin{lemma}\label{lem:generaldp:app}
Fix a randomized algorithm $\mathcal{A}: \mathcal{X} \times R \to \R$ that is $(\ell,\delta)$-sensitive. Then consider the randomized laplace mechanism $\overline{\mathcal{A}}$ which on input $X$ outputs $\mathcal{A}(X,r) + \text{Lap}(0,\ell/\eps)$ where $r \sim R$ is uniformly random string. Then the algorithm $\overline{\mathcal{A}}$ is $(\eps,2(1+e^\eps) \delta )$-differentially private.
\end{lemma}
\begin{proof}
Fix any neighboring datasets $X,X' \in \mathcal{X}$. Let $R^* = R_X \cap R_{X'}$ where $R_X,R_{X'}$ are in Definition \ref{def:generalsensitive:app}. Since$|R_X| > (1-\delta)|R|$ and $|R_{X'}| > (1-\delta)|R|$, we have $|R_X \cap R_{X'}| > (1-2\delta)|R|$. Now fix any $r \in R^*$. By Definition \ref{def:generalsensitive:app}, we know that 
$|\mathcal{A}(X,r) - \mathcal{A}(X',r)| < \ell$.

From here, we follow the standard proof of correctness of the Laplacian mechanism by bounding the ratio 
\begin{align*}
\frac{\Pr[\mathcal{A}(X,r)  + \text{Lap}(0,\frac{\ell}{\eps} ) = z]}{\Pr[\mathcal{A}(X',r)  + \text{Lap}(0,\frac{\ell}{\eps} ) = z]}
\end{align*}
for any $z \in \R$.

In what follows, set $b = \frac{\ell}{\eps}$
\begin{align*}
        & ~\frac{\Pr[\mathcal{A}(X,r)  + \text{Lap}(0,b)= z]}{\Pr[\mathcal{A}(X',r)  + \text{Lap}(0,b)= z]} \\ 
        = & ~ \frac{\Pr[\text{Lap}(0,b)= z - \mathcal{A}(X,r) ]}{\Pr[  \text{Lap}(0,b)= z - \mathcal{A}(X',r) ]}\\
        = & ~ \frac{ \frac{1}{2b}\exp(- {|z-\mathcal{A}(X,r)|}/{b} )}{\frac{1}{2b}\exp(-{|z-\mathcal{A}(X',r)|}/{b}) }\\
        = & ~ \exp\left( ( |z-\mathcal{A}(X',r)| - |z-\mathcal{A}(X,r)| )  / b \right) \\
        \leq & ~ \exp\left(  |\mathcal{A}(X,r) -\mathcal{A}(X',r)| / b \right) \\
        \leq & ~ \exp( \ell / b ) \\
        \leq & ~ e^\eps ,
\end{align*}
where the forth step follows from triangle inequality $|x|-|y|\leq |x-y|$, the last step follows from $\ell/b=\eps$.

 It follows that for any set $S \subset \R$ and any $r \in R^*$, we have 
 \begin{align*}
 &~\Pr\left[\mathcal{A}(X,r)+ \text{Lap}(0,b) \in S\right]  \\
 \leq&~ e^\eps \cdot \Pr\left[\mathcal{A}(X',r) + \text{Lap}(0,b)\in S\right] ,
 \end{align*}
 where the randomness is taken over the generation of the Laplacian random variable $\text{Lap}(0,b)$. Since this holds for all $r \in R^*$,
 in particular it holds for a random choice of $r \in R^*$, thus we have
 \begin{align}\label{eq:dp_distinct_eq_1:app}
      & ~ \Pr_{Z \sim \text{Lap}(0,b), r \sim R^*} \left[\mathcal{A}(X,r)+ Z \in S\right] \notag \\
      & ~ \leq e^\eps \cdot   \Pr_{Z \sim \text{Lap}(0,b), r \sim R^*} \left[\mathcal{A}(X',r) + Z\in S\right]  
\end{align}
Now since $|R^*| \geq (1-2\delta)|R|$, by the law of total probability we have 
\begin{align}\label{eq:dp_distinct_eq_2:app}
     & \Pr_{Z \sim \text{Lap}(0,b), r \sim R}  \left[\mathcal{A}(X,r)+ Z \in S\right]  \notag \\
     = &  \Pr_{Z \sim \text{Lap}(0,b), r \sim R^*} \left[\mathcal{A}(X,r)+ Z \in S\right] \cdot \Pr[ r \in R^* ] \notag \\ 
     &+  \Pr_{Z \sim \text{Lap}(0,b), r \sim R\setminus R^*} \left[\mathcal{A}(X,r)+ Z \in S\right] \cdot \Pr [ r \notin R^* ]  \notag \\
    < &   \Pr_{Z \sim \text{Lap}(0,b), r \sim R^*} \left[\mathcal{A}(X,r)+ Z \in S\right] \notag \\ 
     &+  \Pr_{Z \sim \text{Lap}(0,b), r \sim R\setminus R^*} \left[\mathcal{A}(X,r)+ Z \in S\right] \cdot 2\delta  \notag  \\
    \leq &  \Pr_{Z \sim \text{Lap}(0,b), r \sim R^*} \left[\mathcal{A}(X,r)+ Z \in S\right] + 2\delta 
\end{align}
Similarly, it follows that 
\begin{align}\label{eq:dp_distinct_eq_3:app}
     & ~ \Pr_{Z \sim \text{Lap}(0,b), r \sim R}  \left[\mathcal{A}(X',r)+ Z \in S\right] \notag \\     
      > & ~  \Pr_{Z \sim \text{Lap}(0,b), r \sim R^*} \left[\mathcal{A}(X',r)+ Z \in S\right] (1-2\delta) \notag \\ 
      \geq & ~ \Pr_{Z \sim \text{Lap}(0,b), r \sim R^*} \left[\mathcal{A}(X',r)+ Z \in S\right] - 2\delta 
\end{align}
where the last step follows from probability $\Pr[] \leq 1$.

Combining Eq.~\eqref{eq:dp_distinct_eq_1:app}, \eqref{eq:dp_distinct_eq_2:app} and \eqref{eq:dp_distinct_eq_3:app}, we have 
    \begin{align*}
     & ~ \Pr_{Z \sim \text{Lap}(0,b), r \sim R}  \left[\mathcal{A}(X,r)+ Z \in S\right] \\
     \leq & ~ \Pr_{Z \sim \text{Lap}(0,b), r \sim R^*} \left[\mathcal{A}(X,r)+ Z \in S\right] + 2\delta  \\
     \leq & ~ e^\eps \cdot   \Pr_{Z \sim \text{Lap}(0,b), r \sim R^*} \left[\mathcal{A}(X',r) + Z\in S\right]   + 2\delta \\
     \leq & ~ e^\eps \cdot \Big( \Pr_{Z \sim \text{Lap}(0,b), r \sim R}  \left[\mathcal{A}(X',r)+ Z \in S\right] + 2\delta \Big)  + 2\delta 
    \end{align*}
where the first step follows from Eq.~\eqref{eq:dp_distinct_eq_2:app}, the second step follows Eq.~\eqref{eq:dp_distinct_eq_1:app}, and the last step follows from Eq.~\eqref{eq:dp_distinct_eq_3:app}.

    Now recall that for the actual laplacian mechanism algorithm $\overline{\mathcal{A}}$, for any database $X$ we have
    \begin{align*}
        \Pr [ \overline{\mathcal{A}}(X) \in S ] =  \Pr_{Z \sim \text{Lap}(0,b), r \sim R}  \left[\mathcal{A}(X,r)+ Z \in S\right],
    \end{align*} which complets the proof that $\overline{\mathcal{A}}$ is $(\eps,2(1+e^\eps) \delta )$-differentially private. 
    
\end{proof}

\subsection{Analysis of Distinct Count}\label{sec:analysis_distinct_count:app}

In this section, we thoroughly analyze the properties of Distinct Count~\cite{bjkst02}. We first describe the algorithm in Algorithm~\ref{alg:count:app}. Then we prove its $(\ell, \delta)$-sensitivity and a tighter $(\eps, \delta)$-approximation result (compared with the approximation result in~\cite{bjkst02}).

\begin{algorithm}[!ht]\caption{Distinct Count~\cite{bjkst02}}\label{alg:count:app}
\begin{algorithmic}[1]
\Procedure{\textsc{DistinctCount}}{$I, t$} \Comment{Lemma~\ref{lem:dc-sensitivity:app}}
    \State $d \gets \emptyset$ \Comment{$d$ is a priority-queue of size $t$}
    \For{$x_i \in I$}
        \State $y \gets h(x_i) $ \Comment{$h$: $[m] \rightarrow [0,1]$, is a PRF}
        \If{$|d| < t$} 
        \State $d.\code{push}(y)$
        \ElsIf{$y < d.\code{top}() \; \land \; y \notin d$} 
        \State $d.\code{pop}()$
        \State $d.\code{push}(y)$
        \EndIf
    \EndFor
    \State $v \gets d.\code{top}()$
    \State \Return $t/v$
\EndProcedure
\end{algorithmic}
\end{algorithm}

\noindent\textbf{Sensitivity of distinct count}

\begin{lemma}[Sensitivity of DistinctCount]\label{lem:dc-sensitivity:app}
Assume $r \in R$ is the source of randomness 
of the PRF in DistinctCount (Algorithm~\ref{alg:count:app}), 
where $R \in \{0, 1\}^m$, $n$ is the number of 
distinct element of the input,
for any $16 < t < n/2$, DistinctCount is $(20\log(4/\delta) \frac{n}{t}, \delta)$-sensitive.
\end{lemma}
\begin{proof}
We denote DistinctCount (Algorithm~\ref{alg:count:app}) 
$F: \mathcal{X} \times R \to \R$, and define the same $y_t$ as Section~\ref{sec:tech_lemma}.
Thus, for two neighboring database $X, X' \in \mathcal{X}$ ($\|X\|_0 = n$):

\begin{align*}
    |F(X,r) - F(X',r)| \leq \max \Big\{ \big| \frac{t}{y_t} - \frac{t}{y_{t-1}} \big|,\big| \frac{t}{y_t} - \frac{t}{y_{t+1}} \big| \Big\} .
\end{align*}

{\bf Part 1.}
From second inequality of Lemma~\ref{lem:main} 
(the case $\beta = 1/2$), 
for any $5 < t\leq n/2$ and $4 < \alpha < t/4 $, we have: 
\begin{align*}
    \Pr\left[ \Big| \frac{1}{y_t} - \frac{1}{y_{t+1}} \Big| \leq 4 \alpha \frac{n}{t^2} \right] \geq 1- \exp( - t / 6 ) - \exp( - \alpha / 4 ) 
\end{align*}
It follows that 
\begin{align} \label{eq:dpcount_lem1_1:app}
    & ~ \Pr\left[ \Big| \frac{t}{y_t} - \frac{t}{y_{t+1}} \Big| \leq 5 \alpha \frac{n}{t} \right] \notag \\
    >
    & ~ \Pr\left[ \Big| \frac{t}{y_t} - \frac{t}{y_{t+1}} \Big| \leq 4 \alpha \frac{n}{t} \right] \notag \\  
    = & ~ \Pr\left[ \Big| \frac{1}{y_t} - \frac{1}{y_{t+1}} \Big| \leq 4\alpha \frac{n}{t^2} \right] \notag \\
    \geq & ~ 1- \exp( - t / 6 ) - \exp( - \alpha / 4 ) \notag \\
    \geq & ~ 1- \exp( - (t/4) \cdot (2/3)  ) - \exp( - \alpha / 4 ) \notag \\
    \geq & ~ 1-\exp(- 2 \alpha / 3) - \exp(- \alpha / 4) \notag \\
    \geq & ~ 1 - 2 \exp(-\alpha / 4)
\end{align}
Set $\alpha = 4 \log(4/\delta) $ in Eq.~\eqref{eq:dpcount_lem1_1:app}:
\begin{align*}
    \Pr\left[ \Big| \frac{t}{y_t} - \frac{t}{y_{t+1}} \Big| \leq 20 \log(4/\delta) \frac{n}{t} \right] \geq 1 - \delta/2
\end{align*}

{\bf Part 2.}
Similarly, from the the second inequality of Lemma~\ref{lem:main} (the case $\beta = 1/2$), 
for any $10 < t\leq n/2$ and $4 < \alpha < t/4 $, we have: 
\begin{align*}
    & ~ \Pr\left[ \Big| \frac{1}{y_{t-1}} - \frac{1}{y_t} \Big| \leq 4 \alpha \frac{n}{(t-1)^2} \right] \\
    \geq & ~ 1- \exp( - (t - 1) / 6 ) - \exp ( - \alpha / 4 ) 
\end{align*}
From $t > 16 $, we know $ 0.8t^2 < (t-1)^2$ and $t-1 > 0.75t > 0$. Thus:
\begin{align}\label{eq:dpcount_lem1_2:app}
    & ~ \Pr \left[ \Big| \frac{t}{y_t} - \frac{t}{y_{t-1}} \Big| \leq 5 \alpha \frac{n }{t} \right] \notag \\
    = & ~
    \Pr \left[ \Big| \frac{t}{y_t} - \frac{t}{y_{t-1}} \Big| \leq 4 \alpha \frac{nt}{0.8t^2} \right] \notag \\
    > & ~ \Pr \left[ \Big| \frac{t}{y_t} - \frac{t}{y_{t-1}} \Big| \leq 4 \alpha \frac{nt}{(t-1)^2} \right] \notag\\
    = & ~ \Pr \left[ \Big| \frac{1}{y_{t-1}} - \frac{1}{y_{t}} \Big| 
    \leq 4 \alpha \frac{n}{(t-1)^2} \right] \notag \\
    \geq & ~ 1- \exp( - (t-1) / 6 ) - \exp( - \alpha / 4) \notag \\
    \geq & ~ 1 - \exp( - 0.75t / 6 )  - \exp( - \alpha / 4)  \notag \\
    \geq & ~ 1 - \exp( - \alpha / 2 )  - \exp( - \alpha / 4)  \notag \\
    \geq & ~ 1 - 2 \exp( - \alpha / 4) 
\end{align}
Set $\alpha = 4 \log(4/\delta) $ in Eq.~\eqref{eq:dpcount_lem1_2:app}:
\begin{align*}
    \Pr\left[ \Big| \frac{t}{y_t} - \frac{t}{y_{t+1}} \Big| \leq 20 \log(4/\delta) \cdot \frac{n}{t} \right] \geq 1 - \delta/2
\end{align*}

{\bf Part 3.}
Now apply union bound combining the results of 
{\bf Part 1.} and {\bf Part 2.}. 
Hence, for any $X$, $0 < \delta < 1$, $16 < t < n/2 $:
\begin{align*}
\Pr \left[ |F(X,r) - F(X',r)| \leq 20 \log(4/\delta) \cdot \frac{n}{t} \right] \leq 1 - \delta  .   
\end{align*}
Now, we proved the sensitivity of Algorithm~\ref{alg:count:app}.
\end{proof}

\noindent\textbf{Lemma for approximation guarantees}

\begin{lemma}[Restatement of Lemma~\ref{lem:approximation}]\label{lem:approximation:app}
Let $x_1,x_2,\dots,x_n \sim [0,1]$ be uniform random variables, and let $y_1,y_2,\dots,y_n$ be their order statistics; namely, $y_i$ is the $i$-th smallest value in $\{x_j\}_{j=1}^n$. Fix $\eta \in (0,1/2),\delta \in (0, 1/2)$. Then if $t > 3(1+\eta) \eta^{-2} \log(2/\delta)$, with probability $1-\delta$ we have 
\begin{align*}
    (1-\eta) \cdot n \leq \frac{t}{y_t} \leq (1+\eta) \cdot n .
\end{align*}     
\end{lemma}
\begin{proof}
We define $I_1$ and $I_2$ as follows
\begin{align*}
    I_1 =  ~ [0,\frac{t}{n(1+\eta)} ], ~~~~
    I_2 =  ~ [0,\frac{t}{n(1-\eta)} ].
\end{align*}
First note that if $x \sim [0,1]$, $\Pr[x \in I_1] = \frac{t}{n(1+\eta)}$. Since we have $n$ independent trials, setting $Z = |\{x_i : i \in I_1\}|$ we have $\ex{Z} = \frac{t}{(1+\eta)}$. 

Then by the upper Chernoff bound, we have  
\begin{align*}
    \Pr[ Z > t ] \leq \exp\left(   -\frac{\eta^2 t}{3(1+\eta)} \right) \leq 1- \delta /2 .
\end{align*} 
Similarly, setting $Z' = |\{x_i : i \in I_2\}|$, we have $\ex{Z'} = \frac{t}{(1-\eta)}$, so by the lower Chernoff bound, we have 
\begin{align*}
\Pr[ Z' < t ] \leq \exp ( - \eta^2 t / 2 )  \leq 1 - \delta / 2.
\end{align*}
Thus by a union bound, we have both that $Z < t$ and $Z' > t$ with probability $1 - \delta$. Conditioned on these two events, it follows that $y_t \notin I_1$ but $y_t \in I_2$, which implies that $\frac{t}{n(1+\eta) }< y_t < \frac{t}{n(1-\eta)}$, and so we have 
\begin{align*}
  (1-\eta) n< \frac{t}{y_t} < (1+\eta)n  
\end{align*}
as desired.

\end{proof}

\subsection{Differentially Private Distinct Count}\label{sec:dpdc:app}

\begin{algorithm}\caption{\textsc{DPDistinctCount}: Differentially Private Distinct Count}\label{alg:dpcount:app}
\begin{algorithmic}[1]
\Procedure{\textsc{DPDistinctCount}}{$I, \eps, \eta, \delta$} \Comment{Theorem~\ref{thm:main:app}}
    \State $\textsc{pQueue} \gets \emptyset$ \Comment{$\textsc{pQueue}$ is a priority-queue of size $t$}
    \Comment{ $t \geq \max\big(3(1+ {\eta}/{4})({\eta}/{4})^{-2} \cdot \log({6}/{\delta}),\; 20 \eps^{-1} ({\eta}/{4})^{-1} \cdot \log({24(1+e^{-\eps})}/{\delta}) \cdot \log({3}/{\delta}) )$}
    \For{$x_i \in I$}
        \State $y \gets h(x_i) $ \Comment{$h$: $[m] \rightarrow [0,1]$, is a PRF}
        \If{$|\textsc{pQueue}| < t$} 
        \State $\textsc{pQueue}.\code{push}(y)$
        \ElsIf{$y < \textsc{pQueue}.\code{top}() \; \land y \notin \textsc{pQueue}$} 
        \State $\textsc{pQueue}.\code{pop}()$
        \State $\textsc{pQueue}.\code{push}(y)$
        \EndIf
    \EndFor
    \State $v \gets \textsc{pQueue}.\code{top}()$
    \State $\vldbrevision{\tilde{G}} \leftarrow (1+\frac{3}{4} \eta ) \frac{t}{v} + \text{Lap}(20 \eps^{-1}\frac{n}{t}\log(24(1+e^{-\eps})/\delta))$
    \State \Return \vldbrevision{$\tilde{G}$}
\EndProcedure
\end{algorithmic}
\end{algorithm}

\begin{theorem}[main result. Restatement of Theorem~\ref{thm:main}]\label{thm:main:app}
For any $0 < \eps < 1$, $0 < \eta < 1/2$, $0 < \delta < 1/2$, there is an distinct count algorithm (Algorithm~\ref{alg:dpcount:app}) such that:
\begin{enumerate}
    \item The algorithm is $(\eps, \delta)$-differentially private.
    \item  With probability at least $1 - \delta$, the estimated distinct count $\vldbrevision{\tilde{G}}$ satisfies: 
    \begin{align*}
        n \leq \vldbrevision{\tilde{G}} \leq (1+\eta) \cdot n,
    \end{align*}
    where $n$ is the number of distinct elements in the data stream.
\end{enumerate}
The space used by the distinct count algorithm is \begin{align*}
O\Big( ( \eta^{-2} + \eps^{-1} \eta^{-1}  \log(1/\delta)) \cdot \log(1/\delta) \cdot \log n \Big)
\end{align*}
bits.
\end{theorem}

\begin{proof}
Let $\tilde{F}_0$ be the result output by the original algorithm (Algorithm~\ref{alg:count:app}) with the same $t$.
Our differentially private distinct count algorithm (Algorithm~\ref{alg:dpcount:app}) essentially output $\tilde{A} = (1+ \frac{3}{4}\eta)\tilde{F}_0 + \text{Lap}(\ell/\eps)$, where: 
\begin{align*}
    \ell = & ~ 20 \frac{n}{t} \log(24(1+e^{-\eps})/\delta)  \\
 t = & ~ \max \Big\{ 3(1+{\eta}/{4})({\eta}/{4})^{-2}\log(6/\delta),\;  \\ 
 & \quad \quad \quad  20  \eps^{-1} ({\eta}/{4})^{-1} \cdot \log(24(1+e^{-\eps})/\delta) \cdot \log(3/\delta) \Big\}
\end{align*}

From Lemma \ref{lem:dc-sensitivity:app}, we know that 
for any $16 < t < n/2$,
the original distinct count algorithm (Algorithm~\ref{alg:count:app})
\begin{align*}
\Big( 20\log(4/\delta) \cdot \frac{n}{t} ~~,~~ \delta \Big)-\text{sensitive}.
\end{align*}

After rescaling $\delta$ by a constant factor, it can be rewritten as 
\begin{align*}
    \Big( 20\log(24(1+e^{-\eps})/\delta) \cdot \frac{n}{t} ~~,~~  \frac{\delta}{6(1+e^{-\eps})} \Big)-\text{sensitive}.
\end{align*}
Then, from Lemma \ref{lem:generaldp:app}, we know that it becomes 
$(\eps, \delta/3)$-DP 
by adding $\text{Lap} (\ell/\eps)$ noise when outputting the estimated count, where $\ell$ is as defined above, which completes the proof of the first part of the Theorem.

Next, from Lemma~\ref{lem:approximation:app} and the fact that $t \geq 3(1+{\eta}/{4}) \cdot ({\eta}/{4})^{-2} \cdot \log(6/\delta)$, 
we know that $(1 - \frac{\eta}{4}) n < \tilde{F}_0 < (1 + \frac{\eta}{4}) n$ with probability at least $1-\delta/3$.

Since $(1-\frac{1}{4}\eta)(1+\frac{3}{4}\eta) \leq 1 + \frac{1}{4} \eta $ for any $0 < \eta \leq 1$, we get:
\begin{align*}
\Pr\Big[ (1+\frac{1}{4}\eta) \cdot n <(1+\frac{3}{4}\eta)\tilde{F}_0 < (1+\frac{1}{2}\eta) \cdot n \Big] \geq 1 - {\delta}/{3}
\end{align*}

Next, using the exponential $\Theta(e^{-x})$ tails of the Laplace distribution, and the fact that $\ell = \Omega(\log(1/\delta) \frac{n}{t})$ and $t = \Omega(   \eps^{-1} \eta^{-1} \log^2(1/\delta) )$, 
we have:
\begin{align*}
\Pr \Big[ |\text{Lap}(\ell/\eps)| > \frac{1}{4}\eta n \Big] \leq \delta / 3 .    
\end{align*}
Conditioned on both event that 
\begin{align*}
    |\text{Lap}(\ell/\eps)| < \frac{1}{4}\eta n \text{~~~and~~~} \\
    (1+\frac{1}{4}\eta) \cdot n <(1+\frac{3}{4}\eta) \cdot \tilde{F}_0 < (1+\frac{1}{2}\eta) \cdot n
\end{align*} 
which hold together with probability $1-\delta$ by a union bound, it follows that the estimate $\tilde{A}$ of the algorithm indeed satisfies $n \leq \tilde{A} \leq (1+\eta)n$, which completes the proof of the approximation guarantee in the second part of the Theorem. Finally, the space bound follows from the fact that the algorithm need only store the identities of the $t$ smallest hashes in the data stream, which requires $O(t \log  n)$ bits of space, yielding the bound as stated in the theorem after plugging in 
\begin{align*}
    t = \Theta \Big( (\eta^{-2} + \epsilon^{-1} \eta^{-1} \log(1/\delta) ) \cdot \log(1/\delta) \Big) .
\end{align*}
Thus, we complete the proof.
\end{proof}

\begin{algorithm}\caption{\textsc{$1.1$-approx. DPDistinctCount}}\label{alg:dpcount2:app}
\begin{algorithmic}[1]
\Procedure{\textsc{$1.1$-Approx. DPDC}}{$I, \eps, \delta$} \Comment{Lemma~\ref{lem:approxdp:app}}
    \State $\textsc{pQueue} \gets \emptyset$ \Comment{$\textsc{pQueue}$ is a priority-queue of size $t$}
    \State \Comment{ $t = 10^{3} \eps^{-1} \log({24(1+e^{-\eps})}/{\delta}) \log({3}/{\delta}) )$}
    \For{$x_i \in I$}
        \State $y \gets h(x_i) $ \Comment{$h$: $[m] \rightarrow [0,1]$, is a PRF}
        \If{$|\textsc{pQueue}| < t$} 
        \State $\textsc{pQueue}.\code{push}(y)$
        \ElsIf{$y < \textsc{pQueue}.\code{top}() \; \land y \notin \textsc{pQueue}$} 
        \State $\textsc{pQueue}.\code{pop}()$
        \State $\textsc{pQueue}.\code{push}(y)$
        \EndIf
    \EndFor
    \State $v \gets \textsc{pQueue}.\code{top}()$
    \State $\vldbrevision{\tilde{G}}\leftarrow 1.075  \frac{t}{v} + \text{Lap}(0.02 n / \log(3/\delta) )$
    \State \Return $\vldbrevision{\tilde{G}}$
\EndProcedure
\end{algorithmic}
\end{algorithm}

\begin{claim}[Restatement of Claim~\ref{cla:t-max}]\label{cla:t-max:app}
For any $0 < \delta \leq 10^{-3}$, $0.1  \leq \eta < 1$ 
and $0 < \eps < 1$, then we have 
\begin{align*}
    & ~ 3( 1 + \eta / 4 ) \cdot ( \eta / 4 )^{-2} \cdot \log(6/\delta) \\
     & ~ \leq 25  \eps^{-1} ( \eta/ 4 )^{-1} \cdot \log(24(1+e^{-\eps})/\delta) \cdot \log(3/\delta) .
\end{align*}     
\end{claim}
\begin{proof}
From $0 < \delta < 1$, we know:
\begin{align*}
    \log(6/\delta) \leq 2\log(3/\delta)
\end{align*}
It follows:
\begin{align*}
\LHS \leq & ~ 3\big[(1+{\eta}/{4}) \cdot ({\eta}/{4})^{-1}\big]  \cdot ({\eta}/{4})^{-1} \cdot 2\log(3/\delta) \\
    \leq & ~ 6 ( 4 / \eta + 1 )\cdot ( {\eta}/{4})^{-1} \cdot \log(3/\delta) 
\end{align*}
From $ \eta \geq 0.1 $, we know $4/\eta + 1 \leq 41$. 
From $\delta \leq 10^{-3}$, we also know $\log(24/\delta) \geq 10$. Thus:
\begin{align*}
   LHS \leq & ~ 246({\eta}/{4})^{-1}\cdot \log(3/\delta) \\
   \leq & ~ 25\cdot 10 \cdot ({\eta}/{4})^{-1}\cdot \log(3/\delta) \\
   \leq & ~ 25\log(24/\delta)\cdot ({\eta}/{4})^{-1}\cdot \log(3/\delta) \\
    \leq & ~ 25\log(24(1+e^{-\epsilon})/\delta)\cdot ({\eta}/{4})^{-1}\cdot \log(3/\delta) \\
    \leq & ~ 25 \eps^{-1} ({\eta}/{4})^{-1} \cdot \log(24(1+e^{-\epsilon})/\delta)\cdot  \log(3/\delta) 
\end{align*}
Now we completes the proof.
\end{proof}

\begin{lemma}[Restatement of Lemma~\ref{lem:approxdp}]\label{lem:approxdp:app}
For any $0 < \eps < 1$, $0 < \delta \leq 10^{-3}$, there is an distinct count algorithm (Algorithm~\ref{alg:dpcount2:app}) such that:
\begin{enumerate}
    \item The algorithm is $(\eps, \delta)$-differentially private.
    \item  With probability at least $1 - \delta$, the estimated distinct count $\vldbrevision{\tilde{G}}$ satisfies: 
    \begin{align*}
        n \leq \vldbrevision{\tilde{G}} \leq 1.1 n,
    \end{align*}
    where $n$ is the number of distinct elements in the data stream.
\end{enumerate}
The space used by the distinct count algorithm is \begin{align*}
O\Big( ( 100 + 10 \eps^{-1} \log(1/\delta)) \cdot \log(1/\delta) \cdot \log n \Big)
\end{align*}
bits.
\end{lemma}

\begin{proof}
This lemma directly follows Theorem~\ref{thm:main:app} by setting $\eta = 0.1$, $0 < \delta < 10^{-3}$ and:
\begin{align*}
t & = ~ 25 \eps^{-1} ({\eta}/{4})^{-1} \cdot \log(24(1+e^{-\epsilon})/\delta)\cdot  \log(3/\delta) \\
& = ~ 10^{3} \eps^{-1} \log(24(1+e^{-\epsilon})/\delta)\cdot  \log(3/\delta)
\end{align*}
Thus, the scale factor in the Lap distribution (line 14 in Algorithm~\ref{alg:dpcount:app}) becomes:
\begin{align*}
     & ~ 20 \eps^{-1}\frac{n}{t}\log(24(1+e^{-\eps})/\delta) \\
     = & ~ ~
    20 \eps^{-1} \frac{n \log(24(1+e^{-\eps})/\delta)} {1000 \eps^{-1} \log(24(1+e^{-\epsilon})/\delta)\cdot  \log(3/\delta)} \\
    = & ~ 0.02 n /\log(3/\delta)
\end{align*}
\end{proof}

%% file: rand_partition.tex
\section{Properties of Binomial Distribution}
\label{sec:rand_partition}

\begin{fact}[Tail bounds of binomial distribution]\label{fact:binomial}
If $X \sim B(n, p)$, that is, $X$ is a binomially distributed random variable, where $n$ is 
the total number of experiment and $p$ is the probability of each experiment getting a
successful result, and $k \geq np$, then:
\begin{align*}
    \Pr[ X \geq k ] \leq \exp ( -2n(1 - p - (n-k)/{n})^2 )
\end{align*}

\end{fact}
\begin{proof}
For $k \leq n p$, from the lower tail of the CDF of binomial distribution 
$F(k; n, p) = \Pr[X \leq k]$, we use Hoeffding's inequality~\cite{h63} to get a simple bound:
\begin{align*}
    F(k;n,p) \leq \exp ( -2n (p - {k}/{n})^2 )
\end{align*}

For $k \geq n p$, since $\Pr[X \geq k] = F(n-k; n, 1-p)$, we have:
\begin{align*}
    \Pr[ X \geq k] \leq \exp ( -2n(1 - p - (n-k)/{n})^2 ).
\end{align*}
\end{proof}

\begin{lemma}
If $X \sim B(n, p)$, that is, $X$ is a binomially distributed random variable, where $n$ is 
the total number of experiment and $p$ is the probability of each experiment getting a
successful result, then
\begin{align*}
    \Pr[ X \geq np + \sqrt{ 0.5 n \log (1/\delta)}] \leq \delta
\end{align*}
\end{lemma}

\begin{proof}
From Fact~\ref{fact:binomial}, let $k = np +  \sqrt{0.5n \log (1/\delta) } $, we have:
\begin{align*}
    & ~ \Pr [ X \geq np + \sqrt{0.5 n \log (1/\delta) }
    ] \\
    \leq & ~ \exp (-2n ( 1 - p - 
    (n - (np + \sqrt{- 0.5 n \log(1/ \delta ) } ) ) / n )^2 ) \\
    = & ~ \exp ( -2n \frac{ 1}{2n} \log (1/\delta) ) \\
    = & ~ \delta .
\end{align*}
\end{proof}